\documentclass[notitlepage,superscriptaddress, 11pt]{revtex4-1}

\usepackage{graphicx,graphics,epsfig,times,bm,bbm,amssymb,amsmath,amsfonts,amsthm,mathrsfs,MnSymbol,fullpage,subfig}
\usepackage[matrix,frame,arrow]{xypic}
\usepackage[normalem]{ulem}
\usepackage{lineno}
\usepackage{color}
\usepackage[usenames,dvipsnames,svgnames,table]{xcolor}
\usepackage{enumerate}
\usepackage{relsize}
\usepackage[titletoc]{appendix}

\usepackage{caption}
\newtheorem{thm}{Theorem}

\newtheorem{prop}{Proposition}

\newtheorem{definitionenv}{Definition}
\newtheorem{remarkenv}[definitionenv]{Remark}

\newcommand{\bes} {\begin{subequations}}
\newcommand{\ees} {\end{subequations}}
\newcommand{\bea} {\begin{eqnarray}}
\newcommand{\eea} {\end{eqnarray}}

\newcommand{\defeq}{\stackrel{\mathsf{def}}{=}}

\newcommand{\abs}[1]{\ensuremath{\left|#1\right|}} %absolute value
\newcommand{\norm}[1]{\ensuremath{\left\|#1\right\|}} %norm
\newcommand{\beq}{\begin{equation}}

\newcommand{\eeq}{\end{equation}}

\newcommand{\ip}[2]{\left\langle #1, #2\right\rangle} %absolute value

\def\>{\rangle}
\def\<{\langle}

\newcommand{\braket}[2]{\langle #1\vert #2\rangle}

\newcommand{\ket}[1]{\left|#1\right\rangle}

\usepackage{verbatim}

\def\+#1{\mathcal{#1}}
\newcommand{\bpm}{\begin{pmatrix}}
\newcommand{\epm}{\end{pmatrix}}

\newcommand{\ra}{\rightarrow}

\newcommand{\lp}{\left(}
\newcommand{\rp}{\right)}

\newcommand{\mb}{\mathbb}

\DeclareMathOperator{\polylog}{polylog}

\newtheorem{sub}{Subroutine}
\newenvironment{subroutine}
	{
    	\par\addvspace{\baselineskip}
    	\hrule
        \begin{sub}\rm
   }
   {	
   		\end{sub}
        \hrule\
        \par\addvspace{\baselineskip}
    }

\begin{document}

\title{Quantum algorithms for feedforward neural networks}

\author{Jonathan Allcock}
\email{jonallcock@tencent.com}
\affiliation{Tencent Quantum Laboratory}
\author{Chang-Yu Hsieh}
\email{kimhsieh@tencent.com}
\affiliation{Tencent Quantum Laboratory}
\author{Iordanis Kerenidis}
\email{jkeren@irif.fr}
\affiliation{CNRS, IRIF, Universit\'e Paris Diderot, Paris, France}
\author{Shengyu Zhang}
\email{shengyzhang@tencent.com}
\affiliation{Tencent Quantum Laboratory}
\affiliation{The Chinese University of Hong Kong}

%\linenumbers

%%%%%%%%
\begin{abstract}
    Quantum machine learning has the potential for broad industrial applications, and the development of quantum algorithms for improving the performance of neural networks is of particular interest given the central role they play in machine learning today. In this paper we present quantum algorithms for training and evaluating feedforward neural networks based on the canonical classical feedforward and backpropagation algorithms.  Our algorithms rely on an efficient quantum subroutine for approximating the inner products between vectors in a robust way, and on implicitly storing large intermediate values in quantum random access memory for fast retrieval at later stages.   The running times of our algorithms can be quadratically faster in the size of the network than their standard classical counterparts since they depend linearly on the number of neurons in the network, as opposed to the number of connections between neurons as in the classical case. This makes our algorithms suited for large-scale, highly-connected networks where the number of edges in the network dominates the classical algorithmic running time. Furthermore, networks trained by our quantum algorithm may have an intrinsic resilience to overfitting, as the algorithm naturally mimics the effects of classical techniques such as drop-out used to regularize networks. Our algorithms can also be used as the basis for new quantum-inspired classical algorithms which have the same dependence on the network dimensions as their quantum counterparts, but with quadratic overhead in other parameters that makes them relatively impractical.  
\end{abstract}

\maketitle

\section{Introduction}
Machine learning has been one of the great success stories of computation in recent times. From self-driving cars and speech recognition to cancer detection and product recommendation, the applications of machine learning have had a transformative impact on our lives.  At the heart of many machine learning tasks are artificial neural networks: groups of interconnected computational nodes loosely modelled on the neurons in our brains. In a supervised learning scenario, a neural network is first trained to recognize a set of labelled data by learning a hierarchy of features that together capture the defining characteristics of each label. Once trained, the network can then be evaluated on previously unseen data to predict their corresponding labels.  While their origins can be traced back to the 1940s, artificial neural networks enjoyed a resurgence of interest in the 1980s when a method for computing gradients known as backpropagation \cite{rumelhart1986learning} was popularized, which dramatically increased the efficiency with which networks could be trained to recognize data.  The following decades have seen researchers develop ever more sophisticated techniques to improve the performance of neural networks, and the best algorithms can now outperform humans on realistic image recognition tasks \cite{he2016deep}. Regardless of the additional techniques used, backpropagation remains the key algorithmic component for neural network training.

In spite of the progress made over the years, significant resources are still required to train deep networks needed for industrial and commercial grade problems, and long training times on clusters of GPUs are often necessary.  Furthermore, a network trained on a particular data set for a particular task may not perform well on another task. For this, an entirely different network may need to be trained from scratch.  There is thus considerable benefit in any method for speeding up or otherwise improving the training and evaluation of neural networks.

%Feedforward neural networks play a key role in machine learning, with applications ranging from computer vision and speech recognition to data compression and recommendation systems.  In a supervised learning scenario, a network is trained to recognize a set of labelled data by learning a hierarchy of features that together capture the defining characteristics of each label. Once trained, the network can then be used to recognize unlabelled data.  

The developments in machine learning have been accompanied by the rise of another potentially transformative technology: quantum computing. From the original proposal of Feynman \cite{feynman1982simulating} in the 1980s for using the rules of quantum mechanics to carry out computation, quantum computing has developed into one of the most exciting fields of research today.  The last several years have seen increased attention on applications of quantum computing for industry, of which quantum chemistry, material design, optimization and quantum machine learning are leading candidates.  A natural and fundamentally important question at the intersection of machine learning and quantum computing is whether quantum algorithms can offer any improvement on the classical algorithms currently used for neural networks. 

Several challenges present themselves immediately.  Firstly, the process of training and evaluating neural networks is highly sequential. Quantum algorithms can be a natural fit for performing tasks in parallel by harnessing the phenomenon of quantum superposition. However, they are poorly suited to situations where data must be computed and stored at many intermediate steps, a process which destroys quantum coherence.  Secondly, neural network performance relies critically on the ability to perform non-linear transformations of the data.  Quantum mechanics, however, is linear, and effectively implementing non-linear transformations can be non-trivial \cite{cao2017quantum}.  Finally, if a quantum algorithm is to be used for training a classical neural network, the training data as well as the parameters which define the network must first be encoded in quantum states.  Unless there is an efficient method for preparing these quantum states, the state-preparation procedure can be a time-consuming bottleneck preventing a quantum algorithm from running more quickly than classically. This is especially challenging for the parameters corresponding to the network weights which are so numerous as to rule out efficient methods for generating their corresponding quantum states directly.

There is some qualitative cause for optimism. The standard classical neural network algorithms typically employed make heavy use of linear algebra, for which quantum algorithms may have an advantage in certain cases \cite{biamonte2017quantum}. Additionally, it is desirable for neural networks to be robust to noise and small errors, which can achieved classically by introducing perturbations during the training process \cite{srivastava2014dropout}. In quantum computing, the randomness inherent in the outcome of measurements can have the effect of such introduced perturbations and, if harnessed correctly, a quantum algorithm may achieve a natural robustness to noise without additional cost. 
This is important as, in many cases, raw computational speed is not necessarily the main concern in practical machine learning. Other issues such as the size of the data sets, generalization errors, or how
robust the algorithms are to noise and perturbations may indeed be of much greater practical concern.
Nevertheless, there is an explicit connection between speed and performance. In many cases, the bottleneck is the fact that one can only spend, say, a day for training and not a year. With this restriction, one selects a neural network with a given size so that training can be completed in the allotted time. With faster training, larger neural networks can be trained, and more experiments can be conducted for different choices of network architecture, loss function and hyperparameters.  Taken together, these can all lead to better eventual performance.

%There is some qualitative cause for optimism. Firstly, the standard classical algorithms employed typically make heavy use of linear algebra, for which quantum algorithms may have an advantage in certain cases \cite{biamonte2017quantum}. Secondly, and perhaps more interestingly, it is desirable for neural networks to be robust to noise and small errors, which can achieved by introducing noise during the training process \cite{srivastava2014dropout}. In quantum computing, the effect of introduced noise can occur naturally. In this light, neural networks are a natural fit for quantum computing. On the other hand, neural networks are highly sequential, and data is often required to be measured and stored at many intermediate steps, a process which destroys quantum coherence.  Furthermore, classical data must first be encoded in quantum states if they are to be evaluated by a quantum algorithm, and in some cases this state-preparation procedure may be a bottleneck preventing a quantum algorithm from running more quickly than classically. 

In this work we consider quantum algorithms for feedforward neural networks, in which the propagation of information proceeds from start to finish, layer by layer, without any loops or cycles. Feedforward neural networks not only constitute a canonical category of neural network of fundamental importance in their own right, they also form key components in practical machine learning architectures such as deep convolutional networks and autoencoders. We propose quantum algorithms for training and evaluating robust versions of classical feedforward neural networks---which we call $(\epsilon,\gamma)$-feedforward neural networks---based on the standard classical algorithms for feedforward and backpropagation, adapted to overcome the challenges and exploit the opportunities mentioned above.   To achieve this, we use the fact that quantum computers can compute approximate inner products of vectors efficiently, and define a robust inner product estimation procedure that outputs such an estimate in a quantum data register. By having inner products stored in register, rather than in the phases of quantum states, we are able to directly implement nonlinear transformations on the data, circumventing the problem of non-linearity entirely. We address the state-preparation problem for short length vectors by storing them in a particular data-structure \cite{kerenidis2017quantum,kerenidis2018gradient, wossnig2018quantum} in quantum random access memory (qRAM) \cite{giovannetti2008quantum, prakash2014quantum}, so that their elements may be queried in quantum superposition, and their corresponding states generated efficiently. For the network weight matrices which are too big to be stored efficiently in this manner, we reconstruct their corresponding quantum states indirectly by superposing histories of shorter length vectors stored in qRAM.

In a recent important theoretical work, Tang \cite{tang2018quantum} showed that the data structure required for fast qRAM-based inner product estimation can also be used to classically estimate inner products.  Quantum algorithms based on such data structures can thus give rise to quantum-inspired classical ones as well \cite{tang2018quantum02, gilyen2018quantum}, i.e. classical algorithms based on their quantum counterparts can be defined with only a polynomial slow-down in running time. However, in practice the polynomial factors can make a difference, and analysis of a number of such algorithms \cite{arrazola2019quantum} shows that care is needed when assessing their performance relative to the quantum algorithms from which they were inspired. The situation is the same with the quantum algorithms we propose here.  New classical algorithms for training and evaluating $(\epsilon,\gamma)$-feedforward neural networks based on our quantum algorithms can indeed be defined, and in a later section we compare these with both their quantum versions as well as the standard classical algorithms.  While interesting from a theoretical perspective, these quantum-inspired algorithms will always have a worse performance than their quantum versions, and it is unclear whether they will ever perform well enough in practice to replace existing methods.

Quantum machine learning \cite{schuld2015introduction, biamonte2017quantum, dunjko2018machine} in general and quantum supervised learning \cite{lloyd2013quantum, rebentrost2014quantum,wiebe2015quantum, liu2015fast, schuld2016prediction,qSFA} in particular are fast growing areas of research.  However, while quantum generalizations of neural networks have been proposed for feedforward networks \cite{cao2017quantum, wan2017quantum, romero2017quantum, farhi2018classification, liao2018quantum, schuld2019quantum}, Boltzmann Machines \cite{wiebe2016quantum, verdon2017quantum, kieferova2017tomography, wiebe2019generative} and Hopfield networks \cite{rebentrost2018quantum}, the current work is, to our knowledge, the first proposed quantum algorithms for training and evaluating feedforward neural networks that can, in principle, offer an advantage over the canonical classical feedforward and backpropagation algorithms. While it remains to be seen whether our algorithms can be used in the future to obtain a practical advantage over classical methods, 
we believe, backed with our theoretical analysis and initial simulations, that our results indicate one very promising way that quantum techniques can be applied to neural networks, and that in the future more sophisticated practical techniques may be built on these or other ideas.

\section{Results}
We will discuss $(\epsilon,\gamma)$-feedforward neural networks in more detail later. For now, it is enough to think of these as network in which, with probability at least $1-\gamma$, inner products between vectors can be computed to within error $\epsilon$ of the true value.  Our main results are the following.

\vspace{0.2cm}
\noindent 
{\bf Quantum Training:} [See Theorem \ref{thm:main_result}]
{\em There exists a quantum algorithm for training $(\epsilon,\gamma)$-feedforward neural networks in time $\tilde{O}\lp (TM)^{1.5}N \frac{\log(1/\gamma)}{\epsilon}R\rp$, where $T$ is the number of update iterations, $M$ the number of input samples in each mini-batch, $N$ is the total number of neurons in the network, and $R$ is a factor that depends on the network and training samples, and which numerical evidence suggests is small for many practical parameter regimes.}
\vspace{0.2cm}

 Here and in what follows, $\tilde{O}$ hides polylogarithmic factors in $TMN$. Once the neural network is trained it can be used to label new input data. This evaluation is essentially the same as the feedforward algorithm that is used in the training, and we obtain the following result.

\vspace{0.2cm}
\noindent 
{\bf Quantum Evaluation:} [See Theorem \ref{thm2}]
{\em For a neural network whose weights are explicitly stored in a qRAM data structure, there exists a quantum algorithm for evaluating an $(\epsilon,\gamma)$-feedforward neural network in time $\tilde{O}\lp N \frac{\log(1/\gamma)}{\epsilon}R_e\rp$, where $N$ is the total number of neurons in the network, and $R_e$ is a factor that depends on the network and training samples, and is expected to be small for practical parameter regimes.}
\vspace{0.2cm}

%For a neural network whose weights are already explictly stored in a qRAM data structure or with weights whose corresponding quantum states are easy to construct, the time $T_U$ can be polylogarithmic in $E$, and the total running time $\tilde{O}\lp N\frac{\log(1/\gamma)}{\epsilon}R_e\rp$. 

%For a network with parameters trained via our quantum training algorithm, we store the network weight matrices $W$ in memory only implicitly, which leads to a time $T_U$ scaling as $O(\sqrt{TM})$.  In this case the overall running time equates to $\tilde{O}\lp \sqrt{TM} N \frac{\log(1/\gamma)}{\epsilon}R_e\rp$. 

In contrast, the classical training algorithm has running time $O(TME)$ while the evaluation algorithm takes time $O(E)$, where $E$ is the total number of \emph{edges} in the neural network. It is interesting to note that the complexity order reduces from the number of edges (neural connections) classically to the number of vertices (neurons) in the quantum algorithm. This gap can be very large; indeed, an average neuron in the human brain has 7000 synaptic connections to other neurons \cite{drachman2005we}. 

The running time of our quantum algorithms also depends on $\log(1/\gamma)/\epsilon$ and on terms $R$ and $R_e$ respectively, which depend on the values of certain vector norms which appear during the network evaluation and will be explicated in later sections. While the combined effect of these terms must be taken into account, we will give arguments and numerical evidence that indicate that, in practice, these terms do not contribute significantly to the running time. %In this case, our first quantum algorithm allows for efficiently training large fully-connected neural networks, where the number of edges in the network dominates the running time. 

%%%%%%%%%%%%%%%%%%%%%%%%
\subsection{Feedforward neural networks}
%%%%%%%%%%%%%%%%%%%%%%%%

Feedforward neural networks are covered by a vast amount of literature, and good introductory references are available \cite{goodfellow2016deep, nielsen2017deep}. Here we simply give a brief summary of the concepts required for the current work. 

A feedforward neural network consists of $L$ layers, with the $l$-th layer containing $n_l$ neurons. 
%and define $N=\sum_{l=1}^L n_l$ and $n_{\max} = \max_l n_l$. 
A weight matrix $W^l\in\mb{R}^{n_l\times n_{l-1}}$ is associated between layers $l-1$ and $l$, and a bias vector $b^l\in\mb{R}^{n_l}$ is associated to each layer $l$ (Fig.\ref{fig:neural-network}).  The total number of neurons is $N = \sum_{l=1}^L n_l$ and the total number of edges in the network is $E = \sum_{l=2}^L n_l\cdot n_{l-1}$. For each level $l$,  let $a^l \in\mb{R}^{n_l}$ to be the vector of outputs (activations) of the neurons.
Given a non-linear activation function $f$, the network feedforward rule is given by $a^{l}_j = f(z^l_j)$, where $z^l_j = \sum_k W^l_{jk}a^{l-1}_l + b^l_j$ (Fig \ref{fig:neuron}). It will be convenient to express this as 
\begin{linenomath*}
\begin{equation}\label{eq:network_update}
   z^l_j = \< W^l_j, a^{l-1} \> +b^l_j,  
\end{equation}
\end{linenomath*}
where $\langle x,y\rangle$ is the Euclidean inner product between vectors $x$ and $y$, and $W^l_j$ is the $j$-th row of $W^l$, i.e. $\lp W^l_j\rp_k = W^l_{jk}$.
For simplicity, we assume that the activation function is the same across all layers in the network, although our algorithm works more generally when the activation function is allowed to vary across layers. Common activation functions are the sigmoid $\sigma(z) = \frac{1}{1+ e^{-z}}$ (when neurons take values in $[0,1]$), the hyperbolic function $\tanh(z) = 2\sigma(2z) - 1$ (when neurons take values in  $[-1,1]$), and the Rectified Linear Unit (ReLU function) $f(x) = \max(0,x)$. A training set $\+{T}=\{(x^{1},y^{1}),(x^{2},y^{2}),\ldots \}$ for a learning task consists of vectors $x^{i}\in\mb{R}^{n_{1}}$ and corresponding labels $y^{i}\in\mb{R}^{n_L}$.  This set is used, via a training algorithm, to adjust the network weights and biases so as to minimize a chosen cost function $C:\mb{R}^{n_L}\ra \mb{R}$, which quantifies the network performance. The goal is to obtain parameters such that when the network evaluates a new input which is not part of the training set, it outputs the correct label with a high degree of accuracy.

%\begin{figure}[h]%
%\subfloat[]{\includegraphics[height=4.5cm]{neural-network.pdf}}
%\subfloat[]{\includegraphics[height=4.5cm]{neuron.pdf}}
%\caption{(a) Feedforward neural network with $L=3$ layers. A weight matrix $W^l\in\mb{R}^{n_l\times n_{l-1}}$ is associated between levels $l-1$ and $l$, and a bias vector $b\in\mb{R}^{n_l}$ (not shown) is associated to each level $l$. (b) The $j$-th neuron at level $l$ of a network. The output of the neuron is $f(z_j)$, where $z_j = \sum_k W^l_{jk}a^{l-1}_k + b^l_j$, and $f$ is a given activation function.}
%\label{fig:neural-network}
%\end{figure}

\begin{figure}%
\includegraphics[height=5.0cm]{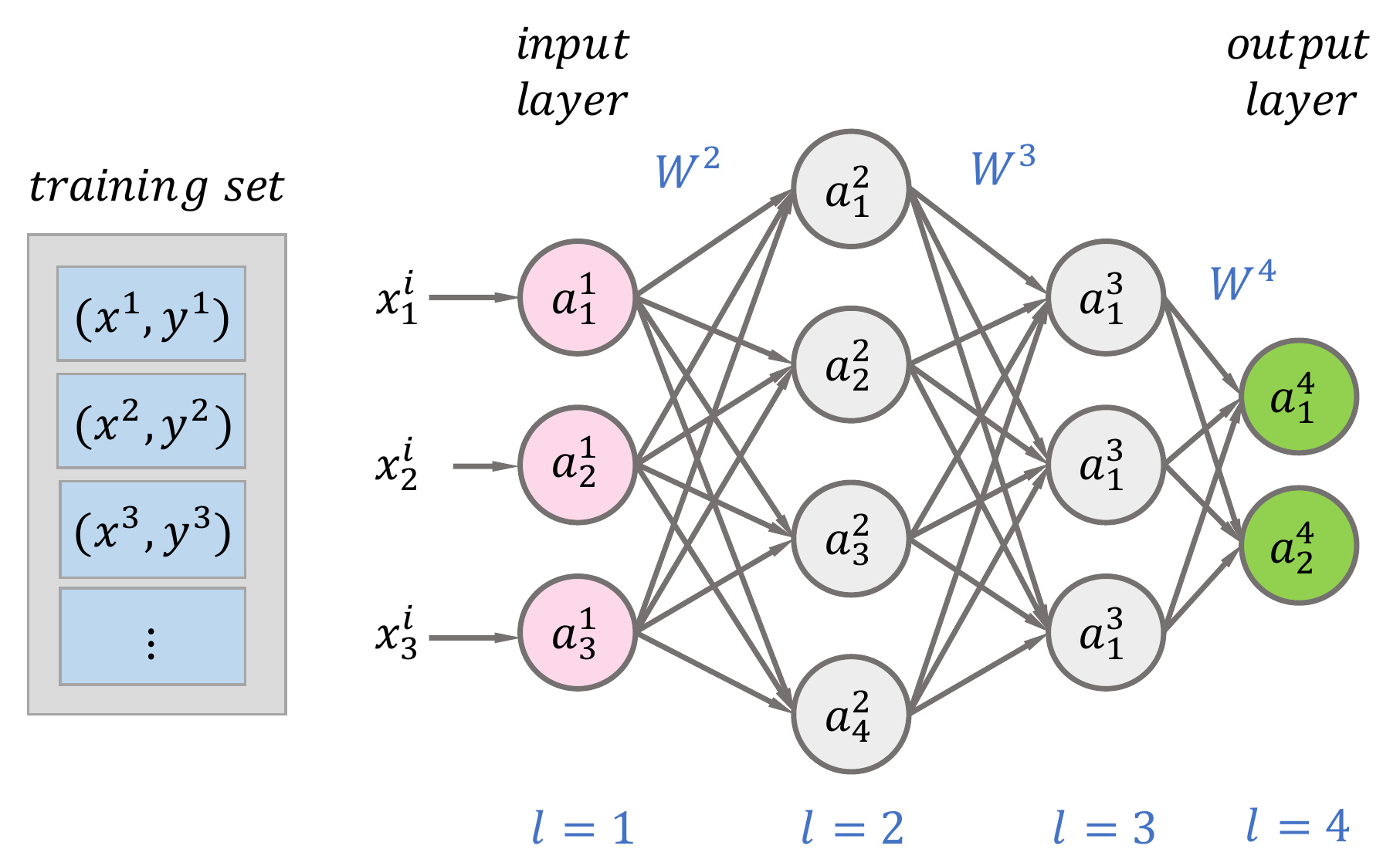}
\caption{Feedforward neural network with $L=4$ layers. A weight matrix $W^l\in\mb{R}^{n_l\times n_{l-1}}$ is associated between levels $l-1$ and $l$, and a bias vector $b\in\mb{R}^{n_l}$ (not shown) is associated to each level $l$. During training, a (data,label) pair $(x^{i},y^{i})$ is chosen from the training set, and the data is fed into the input layer of neurons. The weights and biases in the network determine the activations of subsequent layers of neurons.  Finally, the values from the output layer of neurons is compared with the true label $y^i$.}
\label{fig:neural-network}%
\end{figure}

\begin{figure}%
\includegraphics[height=4.0cm]{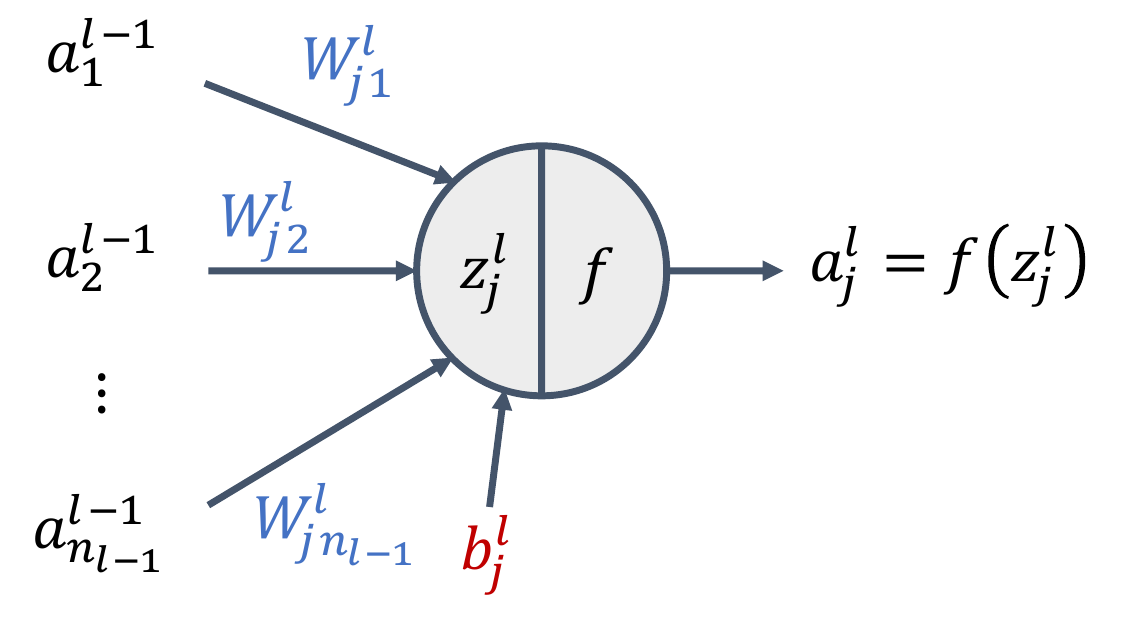}
\caption{The $j$-th neuron at level $l$ of a network. The output activation of the neuron is computed in two steps. Firstly, the value $z_j = \sum_k W^l_{jk}a^{l-1}_k + b^l_j$ is calculated, which depends on the weights $W^l_{ji}$ and bias $b^l_j$ as well as the activations from the previous layer of neurons.  A non-linear activation function $f$ is then applied to $z^l_j$, and the value $f(z_j)$ is output.}
\label{fig:neuron}%
\end{figure}

Classically, network training consists of the following steps:
\begin{enumerate}
\item \textbf{Initialization of weights and biases.} Various techniques are used, but a common choice is to set the biases to a small constant, and to draw the weight matrices randomly from a normal distribution according to $W^{l}_{jk}\sim\+N\lp 0,\frac{1}{\sqrt{n_{l-1}}}\rp$ \cite{glorot2010understanding}. 
\item \textbf{Feedforward.} Select a pair $(x,y)$ from the training set.  Assign the neurons in the first layer of the network to have activations $a^{1}_j = x_j$. Pass through the network layer-by-layer, at each layer $l$ computing and storing the vectors $z^{l} = \left\langle W^{l}_j, a^{l-1}\right\rangle + b^l_j$ and $a^{l} = f\lp z_j^{l}\rp$. The running time of this  procedure is $O\lp \sum_{l=2}^L n_l\cdot n_{l-1}\rp = O\lp E\rp$: at each level $l$ one must evaluate $n_l$ activations, and each activation involves calculating an inner product of dimension $n_{l-1}$ which takes times $n_{l-1}$.
\item \textbf{Backpropagation.}  Given $z^{L}$ and $a^L$ at the end of the feedforward process, define vector $\delta^{L}$ with components $\delta^{L}_j = f'\lp z^{L}_j \rp \frac{\partial C}{\partial a^L_j}$, where $C$ is the chosen cost function. Then, proceeding backwards through the network, compute and store the vectors $\delta^{l}_j = f'\lp z^l_j\rp\left\langle \lp W^{l+1} \rp^T_j, \delta^{l+1}\right\rangle$, where $\lp W^{l+1}\rp^T$ is the matrix transpose of $W^{l+1}$. The running time of the backpropagation algorithm is again $O(E)$, by the same reasoning as for the feedforward algorithm.
\item\textbf{Update weights and biases.} Repeat the feedforward and backpropagation steps for a mini-batch of inputs of size $M$, and then update the weights and biases by stochastic gradient descent, according to:
\begin{linenomath*}
\begin{align}
W^{t+1,l}_{jk} &= W^{t,l}_{jk} - \eta^{t,l} \frac{1}{M}\sum_{m}^Ma^{t,m,l-1}_k \delta^{t,m,l}_j  \label{eq:w_update} \\
b^{t+1,l}_j &= b^{t,l}_j - \eta^{l} \frac{1}{M}\sum_{m}^M \delta^{t,m,l}_j
\end{align}
\end{linenomath*}
where the superscipts $t\in[T]$ and $m\in[M]$ denote the iteration number and mini-batch element respectively, and $\eta^{t,l} >0$ are update step sizes. 
\item \textbf{Iterate}. Repeat the previous weight and bias update step $T$ times, each time with a different mini-batch.  
\end{enumerate}
Note that the method for choosing each mini-batch of inputs is left to the user. A common choice in practice, and the one that used in our numerical simulations, is to divide the training procedure into a number of periods known as \textit{epochs}.  At the beginning of each epoch, the training set is randomly shuffled and successive mini-batches of size $M$ are chosen without replacement.  The epoch ends when all training examples have been selected, after which the next epoch begins.  Alternatively, one may consider selecting the mini-batches by randomly sampling-with-replacement from the full training set. 

The overall running time of the classical training algorithm is $O(TME)$ since there are $TM$ steps and each step requires performing algorithms to carry out the feedforward and backpropagation procedures. Note that the product $TM$ is equal to the number of epochs multiplied by the size $\abs{\+{T}}$ of the training set. The fact that the training of general feedforward neural networks depends linearly on the number of edges makes large-size fully-connected feedforward networks computationally expensive to train. Our quantum algorithm reduces the dependence on the network size from the number of edges to the number of vertices. 

Once the network training is complete, the classical evaluation algorithm consists of running the feedforward step on a previously unseen input, i.e. on data that lies outside the training set. The success of deep learning tells us that the network training produces weights and biases that perform well in classifying new data. 

\subsection{Robust feedforward neural networks}

A perennial concern in neural network training is the problem of overfitting.  Since a large network may have many more free parameters than training data points, it is easy to train a network that recognizes the training data accurately, yet performs poorly on data that lies outside the training set. Various techniques are used in practice to reduce overfitting. For instance, one may consider adding a term to the cost function which penalizes a network with too many large parameters.  Alternatively, one can train the data on an ensemble of different networks, and average the performance across the ensemble. While this may be prohibitively costly to perform exactly, the effect can be approximated by randomly deactivating certain neurons in each layer, a technique known as dropout.  A similar result can be achieved by adding random multiplicative Gaussian noise to the activation of each neuron during training. These techniques are both seen to be effective at preventing overfitting \cite{srivastava2014dropout}, and can also improve network performance by avoiding local minima in parameter space.  

%Indeed, such robust training is similar in spirit to the classical techniques of multiplicative Gaussian noise and dropout \cite{srivastava2014dropout}, where at every layer of the network noise is either introduced in the calculations, or neurons are probabilistically deactivated. Training in this manner can often improve network performance by both avoiding local minima and preventing over-fitting to the training data.

Motivated by the benefits of such random network perturbations during training, and anticipating the requirement of quantum processes to generate randomized outcomes, we consider a generalization of the training and evaluation algorithms where the inner products in the feedforward and backpropagation steps may not be evaluated exactly. Instead, with probability $1-\gamma$, they are estimated to within some error tolerance $\epsilon$, either relative or absolute, depending on whether the inner product is large or small.  That is, the feedforward step computes values $s^l_j$ satisfying 
\begin{linenomath*}
\begin{equation*}
    \abs{s_j^{l} - \ip {W_j^{l}}{a^{l-1}} } \leq \max \left\{ \epsilon  \abs{ \ip{W_j^{l}}{a^{l-1}}}, \epsilon \right\}  \quad \text{with probability } \geq 1-\gamma, 
\end{equation*}
\end{linenomath*}
and similarly for the inner product calculation between $\lp W^{l+1}\rp_j^T$ and $\delta^{l+1}$ in the backpropgation step. Note that we do not specify how the $s^{l}_j$ are generated above. The way this is realized will be left to specific implementations of the algorithm.  For instance, a simple classical implementation is to first compute the inner product $\ip{W_j^{l}}{a^{l-1}}$ and then add independent Gaussian noise bounded by the maximum of $\epsilon \abs{ \ip{ W_j^{l}}{a^{l-1}}}$ and $\epsilon$. In the quantum case, the $s^{l}_j$ will be generated by a quantum inner product procedure which is not perfect, but rather outputs an estimate of the true inner product satisfying the $(\epsilon,\gamma)$ conditions required. The reason we allow for either a relative or absolute error is to ensure that the quantum procedure can be carried out efficiently regardless of the magnitude of the inner products involved.

We refer to such networks as $(\epsilon,\gamma)$-feedforward neural networks, of which the standard classical feedforward neural network is a special case. Our simulation results show that, for reasonable tolerance parameters, the generalization to $(\epsilon,\gamma)$ estimates of inner products does not hurt network performance.   

For small enough $\epsilon$, the running time of classically computing the feedforward and backpropagation steps remains $O(E)$, since in general one needs to look at a large fraction of the coordinates of two vectors to obtain an $\epsilon$-error approximation to their inner product. The classical training time for $(\epsilon,\gamma)$-feedforward neural networks is thus $O(TME)$, as for the original case where inner products are evaluated exactly.

%%%%%%%%%%%%%%%%%%%%%%%%%%%%%%%%%%%%%%%%%%%
\subsection{Quantum Training}
%%%%%%%%%%%%%%%%%%%%%%%%%%%%%%%%%%%%%%%%%%%

While it is clearly desirable to improve on the $O(TME)$ classical result using a quantum algorithm, there are several obstacles to achieving this. As mentioned in the introduction, feedforward neural network training and evaluation is a highly sequential procedure, where at each point one needs to know the results of previously computed steps. Quantum algorithms, by contrast, are typically well suited to performing tasks in parallel, but not for performing tasks which require sequential measurements to be performed, a process which destroys quantum coherence. In addition,  a critical step classically is the application of a non-linear activation function to each neuron. Given that quantum mechanics is inherently linear, applying non-linearity to quantum states is non-trivial. Finally, the size of each weight matrix is $n_l \times n_{l-1}$, so even explicitly writing down these matrices for every step of the algorithm takes time $O(E)$. 

We address these challenges by using a hybrid quantum-classical procedure which follows the classical training algorithm closely.  In our algorithm all sequential steps are taken classically.  At each step, quantum operations are only invoked for estimating the inner products of vectors, and for reading and writing data to and from qRAM.

Given a vector $x \in \mb{R}^n$, define the corresponding normalized quantum state $\ket{x} = \frac{1}{\norm{x}}\sum_{j=1}^n x_j\ket{j}$, where $\norm{\cdot}$ is the $\ell_2$ norm. The inner product of two quantum states therefore satisfies  $\< x | y \> := \frac{\ip{x}{y}}{\norm{x}\norm{y}}$. Two key ingredients are the following:

\medskip\textbf{Robust Inner Product Estimation (RIPE).} If quantum states $\ket{x}$ and $\ket{y}$ can each be created in time $T_U$, and if estimates of the norms $\norm{x}$ and $\norm{y}$ are known to within $\epsilon/3$ multiplicative error, then a generalization of the inner product estimation subroutine of \cite{kerenidis2018qmeans} allows one to perform the mapping $\ket{x}\ket{y}\ket{0}\rightarrow\ket{x}\ket{y}\ket{s}$ where, with probability at least $1-\gamma$:
\begin{linenomath*}
\begin{equation*}
  \abs{s-\< x,y\>} \leq 
  \begin{cases}
   \epsilon  \abs{\< x,y\>} & \text{in time }\widetilde{O}\left(\frac{T_U \log(1/\gamma)}{ \epsilon}\frac{\norm{x}\norm{y}}{ \abs{\ip{x}{y}}}\right)\\
   \epsilon & \text{in time }\widetilde{O}\left(\frac{T_U \log(1/\gamma)}{ \epsilon} \norm{x}\norm{y}\right)
  \end{cases}
\end{equation*}
\end{linenomath*}

The inner product estimates $s$ above are computed in a quantum register and not in the phase of a quantum state, enabling non-linear activation functions to be applied to it.  When the data required is stored in qRAM (see below), this inner product calculation is efficient and provides roughly a factor $O(n_l)$ saving in time per layer of the network. %Both relative error and absolute errors are required to deal with the cases where the inner products are either very large in magnitude or close to zero. If the inner product is small, the running time to achieve an $\epsilon$ relative error is prohibitive. If the inner product is large, the running time for achieving an absolute error becomes costly. In practice, we find that the networks can be well trained even when we achieve a relative error when the inner products are large and an additive error when they are very small.  
    
 \medskip\textbf{Quantum Random Access Memory.} For the RIPE algorithm to efficiently compute an approximation of the inner product $\ip{x}{y}$, we need an efficient way to prepare the states $\ket{x}$ and $\ket{y}$. A qRAM is a device that allows for classical data be queried in superposition. That is, if the classical vector $x\in\mb{R}^N$ is stored in qRAM, then a query to the qRAM implements the unitary $\sum_j \alpha_j\ket{j}\ket{0}\ra \sum_j \alpha_j\ket{j}\ket{x_j}$.  Importantly, if the elements $x_j$ of $x$ arrive as a stream of entries $(j, x_j)$ in some arbitrary order, then $x$ can be stored in a particular data structure \cite{kerenidis2017quantum}---which we will refer to as an $\ell_2$-Binary Search Tree ($\ell_2$-BST)---in a time linear in $N$ (up to logarithmic factors) and, once stored, $\ket{x} = \frac{1}{\norm{x}}\sum_j x_j\ket{j}$ can be created in time polylogarithmic in $N$ (Fig.\ref{fig:qRAM-data-structure}).

\medskip Conceptually, we would like to follow the classical training algorithm, modified so that the network biases $b^{t,l}$, weights $W^{t,l}$, pre-activations $z^{t,m,l}$, activations $a^{t,m,l}$ and backpropagation vectors $\delta^{t,m,l}$  are stored in a qRAM $\ell_2$-BST at every step. Their corresponding quantum states can then be efficiently created, and the RIPE algorithm used to estimate the inner products required to update the network. However, there is a problem.  The weight matrices $W^{t,1}, W^{t,2},\ldots, W^{t,L}$ have a combined total of $E$ entries. Thus it will take time $\tilde{O}(E)$ just to write all the matrix elements into qRAM. This is true even for the initial weight matrices if they are generated by choosing independent random variables for each element, as is common classically. Each time the weights are updated, it will take an additional $\tilde{O}(E)$ time to write the new values to qRAM. Furthermore, the RIPE algorithm requires estimates of the norms of the vectors involved.  This is not a problem for the $a^{t,m,l}$ and $\delta^{t,m,l}$ vectors, as their storage in the $\ell_2$-BST allows their norms to be accessed.  However, if the weight matrices are not stored in this way, then the norms of their rows and columns are also not available to use explicitly. 

We circumvent these issues in two ways.

\medskip\textbf{Low rank initialization.} We set the initial weights $W^{1,l}$ to be low rank by selecting a small number $r$ of pairs of random vectors $a^{0,m,l}$ and $\delta^{0,m,l}$ ($m =1,\ldots ,r$) and taking the sums of their outer-products. If we define $\eta^{0,1}=-1$ then equation \eqref{eq:w_update} for the weights can be expressed as
\begin{linenomath*}
\begin{equation*}
W^{t,l}_{jk} = \sum_{\tau=0}^{t-1}\sum_{m=1}^M \frac{-\eta^{\tau,l}}{M} \delta^{\tau,m,l}_j a^{\tau,m,l-1}_k    
\end{equation*}
\end{linenomath*}
where, for a fixed $t$ and $l$, only $r$ of the $M$ possible $\delta^{0,m,l}$ and $\delta^{0,m,l}$ are non-zero.
We shall see that this low-rank initialization does not affect the network performance in practice, and in section \ref{sec:discussion} we provide a classical justification for this initialization as well as numerical evidence supporting the use of $r =O(\log n_l)$ random pairs.

\medskip\textbf{Implicit storage of weight matrices.}  For each $t\in[T],l\in\{2,\ldots L\},j\in[n_L]$, define the matrix $X^{[t,l,j]}\in\mb{R}^{t\times M}$ with matrix elements $\lp X^{[t,l,j]}\rp_{\tau m} = \frac{-\eta^{\tau,l}}{M}\delta_j^{\tau,m,l}\norm{a^{\tau,m,l-1}}$, with $\tau \in\{0,\ldots t-1\}$, $m\in [M]$ . We will store all the matrix elements of each $X^{[t,l,j]}$ in qRAM (Fig.\ref{fig:X-matrix}), which allows for efficient computation of the states $\ket{W^{t,l}_j}$ on the fly, as opposed to explictly storing all the values $W^{t,l}_{j,k}$ which is prohibitively expensive. More specifically, in the Methods section we show that by querying the rows of each matrix in superposition over iterations $\tau < t$, it is possible to generate the weight states $\ket{W^{t,l}_j}$ in time  $T_W = \tilde{O}\lp \frac{\norm{X^{[t,l,j]}}_{F}}{\norm{W^{t,l}_j}}\sqrt{TM}\rp$, and estimate $\norm{W^{t,l}_j}$ to multiplicative error $\xi$ in time $O(T_W/\xi)$ respectively, where $\norm{X^{[t,l,j]}_F}$ is the Frobenius norm  $X^{[t,l,j]}$ . Similar results apply to generating the states $\ket{\lp W^l\rp^T_j}$ corresponding to the columns of the weight matrices, and estimating their norms.  

\begin{figure}
\includegraphics[height=5.0cm]{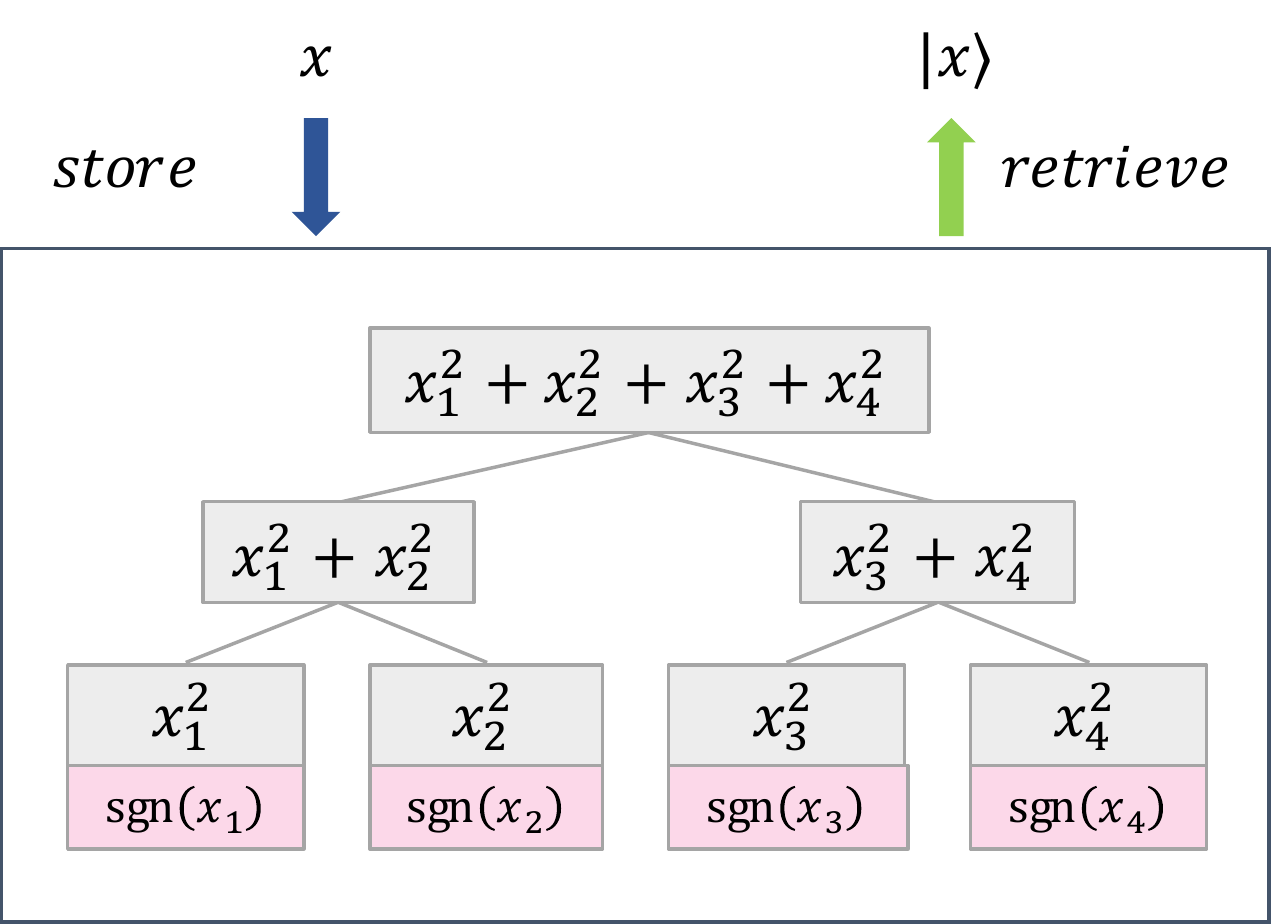}
\caption{$\ell_2$-BST data structure.  A classical vector $x\in\mb{R^N}$ is stored in a binary tree, with leaves containing the squared vector components $x_i^2$ as well as the signs of $x_i$. The tree is populated recursively, with each parent storing the sum of the values of its children. The root of the tree thus stores $\norm{x}$. As shown in \cite{kerenidis2017quantum}, it takes time $\tilde{O}(N)$ to store $x$ in the date structure. Once stored, qRAM access to the data structure allows the quantum state $\ket{x} = \frac{1}{\norm{x}}\sum_j x_j \ket{j}$ to be generated in time polylogarithmic in $N$.}
\label{fig:qRAM-data-structure}
\end{figure}

\begin{figure}
\includegraphics[height=4.0cm]{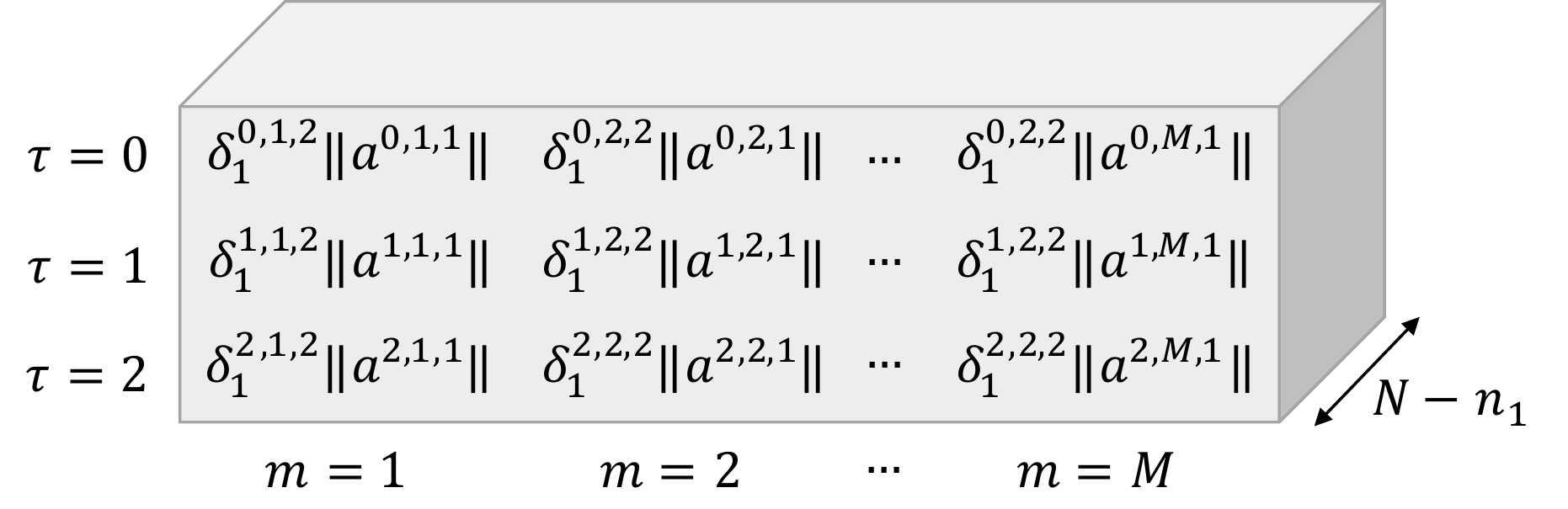}
\caption{Implicit storage of weight matrices via the matrices $X^{[t,l,j]}$. At any iteration $t$ there are $\sum_{l=2}^L n_l = N-n_1$ such matrices, which are stored in the $\ell_2$-BST data structure of Fig.\ref{fig:qRAM-data-structure}. Shown here are the matrix elements of $X^{[3,2,1]}$, with the factors of $-\eta^{\tau,l}/M$ omitted for clarity of presentation. The other  matrices corresponding to different $l,j$ values are indicated behind. Storing these matrices in an $\ell_2$-BST allows for efficient creation of the states $\ket{W^l_j}$ and estimation of their norms.  After each iteration of training another row of data is added to each of $N-n_l$ matrices, and the data structure is updated. The maximum size of data stored occurs after $T$ iterations, and corresponds to $N-n_1$ matrices, each of size $TM$.  Any element of these matrices can thus be coherently accessed by qRAM in time $\polylog(TMN)$. }
\label{fig:X-matrix}
\end{figure}

%Note  that when when we say ``store $a_j^{t,m,l}$ in qRAM'', we mean that the corresponding value is stored as a leaf in a tree data structure and all intermediate nodes of the tree are updated (thus the time to store each element is $\log^2(TMN)$), so that the quantum states corresponding to the normalized vector $a^{t,m,l}$ can be efficiently prepared in time $O(\log(TMN))$ and the norm of the vector is also calculated and stored.  

\medskip With these ideas, one can define a quantum $(\epsilon,\gamma)$-feedforward algorithm which adapts the classical one to make use of RIPE, qRAM access and the implicit storage of the weight matrices. 
%LEt $Z$ a function defined for any two vectors $(x^0,x^1)$ as $Z(x^0,x^1)=\max \{ \norm{x^0},\norm{x^1}\}$.

\begin{subroutine}(\textbf{Quantum }\textbf{$(\epsilon,\gamma)$-Feedforward})\label{sub:q-ff}\\
{\bf Inputs}: indices $t \in [T], m \in [M]$; input pair $(x^{t,m},y^{t,m})$ in qRAM; unitaries $U_{W_j^{t,l},a^{t,m,l-1}}$ for creating $\ket{W^{t,l}_j}$ and $\ket{a^{t,m,l-1}}$ in time $T_U$, estimates of their norms $\overline{\norm{W^{t,l}_j}}$ and $\overline{{\norm{a^{t,m,l-1}}}}$ to relative error at most $\xi = \epsilon/3$, and vectors $b^{t,l}$ in qRAM for $l \in [L]$; activation function $f$; accuracy parameters $\epsilon, \gamma >0$. 
\begin{enumerate}
\item For  $j=1$ to $n_1$ do:
\item \qquad $a^{t,m,1}_j=x^{t,m}_j$
\item For $l = 2$ to $L$ do:
\item \qquad For  $j=1$ to $n_L$ do:
\item \qquad \qquad Use the RIPE algorithm with unitary $U_{W_j^{t,l},a^{t,m,l-1}}$ to compute $s_j^{t,m,l}$, such that \begin{linenomath*}
\begin{equation*}
 \abs{s_j^{t,m,l} - \< W_j^{t,l}, a^{t,m,l-1}\>} \leq \max \{ \epsilon |\< W_j^{t,l}, a^{t,m,l-1}\>|, \epsilon \} \mbox{ with probability} \geq 1-\gamma   
\end{equation*}
\end{linenomath*}
\item  \qquad \qquad Compute $z_j^{t,m,l} = s_j^{t,m,l} + b_j^{t,l}$ and store $z_j^{t,m,l}$ in qRAM
\item  \qquad \qquad Compute $a_j^{t,m,l} = f(z_j^{t,m,l})$ and store $a_j^{t,m,l}$ in qRAM
%\item Output $a^{m,l}$  for all $m \in [M], l \in [L]$.
\end{enumerate}
\end{subroutine}

The cost of performing this feedforward procedure is  $O\lp \sum_{l=2}^L\sum_{j=1}^{n_l}T_{RIPE}(W^{t,l}_j, a^{t,m,l-1})\rp = O\lp N\overline{T}_{RIPE}\rp $, where $T_{RIPE}(W^{t,l}_j, a^{t,m,l-1})$ is the time required to perform Robust Inner Product Estimation between vectors $W^{t,l}_j$ and $a^{t,m,l-1}$, and $\overline{T}_{RIPE}$ denotes the average time (over all layers and neurons) to perform this inner product estimation. Using the running time of RIPE that achieves the larger of the absolute or relative estimation errors, one obtains
\begin{linenomath*}
\begin{equation*}
 T_{RIPE}(W^{t,l}_j, a^{t,m,l-1}) = \widetilde{O}\lp \frac{T_U \log(1/\gamma)}{\epsilon } \frac{\norm{W^{t,l}_j}\norm{a^{t,m,l-1}} }{\max\left\{1, \abs{\ip{W^{t,l}_j}{a^{t,m,l-1}}}\right\}} \rp,   
\end{equation*}
\end{linenomath*}
where $T_U$ is the time required to prepare state $\ket{W^{t,l}_j}$ and $\ket{a^{t,m,l-1}}$. $\ket{a^{t,m,l-1}}$ can be created in time polylogarithmic in $N$ since the required data is, by assumption, stored in qRAM.  By the implicit storage of weight matrices, $\ket{W^{t,l}_j}$ can be created in time $\tilde{O}\lp \frac{\norm{X^{[t,l,j]}}_{F}}{\norm{W^{t,l}_j}}\sqrt{TM}\rp$.  The overall running time of the quantum feedforward algorithm is therefore 
\begin{linenomath*}
\begin{equation*}
 \tilde{O} \lp \sqrt{TM}N  \frac{ \log (1/\gamma)}{\epsilon}R_a^{t,m}\rp,    
\end{equation*}
\end{linenomath*}
where  $R^{t,m}_a = \frac{1}{N-n_1}\sum_{l=2}^L\sum_{j=1}^{n_l}  
\frac{\norm{X^{[t,l,j]}}_F\norm{a^{t,m,l-1}} }{\max\left\{1, \abs{\ip{W^{t,l}_j}{a^{t,m,l-1}}}\right\}}$.  

The factor $R^{t,m}_a$ does not appear in the classical algorithms, and while it is a priori not clear what impact this will have on the running time, we give evidence in the discussion section that this does not impact the running time significantly in practice. The upside is we save a factor of $O(N)$ compared with the classical case, since our algorithm depends linearly on $N$ and not on the number of edges $E$. Last, there is an overhead of $\sqrt{TM}$ which is a consequence of only saving the weight matrices implicitly. For large neural networks one expects that $N >> \sqrt{TM}$.

We can similarly define a quantum backpropagation algorithm:

\begin{subroutine}(\textbf{Quantum} $(\epsilon,\gamma)$-\textbf{Backpropagation})\label{sub:q-bp}\\

\noindent
{\bf Inputs}: indices $t \in [T], m \in [M]$; input pair $(x^{t,m},y^{t,m})$ in qRAM; vectors $a^{t,m,l}, z^{t,m,l}, l\in [L]$ in qRAM; unitaries $U_{(W^{t,l+1})^j,\delta^{t,m,l+1}}$ for creating $\ket{(W^{t,l+1})^j}$ and $\ket{\delta^{t,m,l+1}}$ in time $T_U$, estimates of their norms $\overline{\norm{(W^{t,l+1})^j}}$ and $\overline{{\norm{\delta^{t,m,l+1}}}}$ to relative error at most $\xi = \epsilon/3$, and vectors $b^{t,l}$ in qRAM for $l \in [L]$; derivative activation function $f'$; parameters $\eta^{t,l}$ for $l \in [L]$; accuracy parameters $\epsilon, \gamma >0$. 

\begin{enumerate}
\item For  $j=1$ to $n_L$ do:
\item \qquad $\delta^{t,m,L}_j= f'(z^{t,m,L}_j) \frac{\partial  C}{\partial a^L_j} $.
\item For $l = L-1$ to $1$ do:
\item \qquad For  $j=1$ to $n_l$ do:
\item \qquad \qquad Use the RIPE algorithm with unitary $U_{(W^{t,l})^j,\delta^{t,m,l+1}}$ to compute $s_j^{t,m,l}$, such that  
\begin{linenomath*}
\begin{equation*}
 \abs{s_j^{t,m,l} -  \< \lp W^{t,l+1}\rp^j, \delta^{t,m,l+1}\>} \leq \max \{ \epsilon |\< \lp W^{t,l+1}\rp^j, \delta^{t,m,l+1}\>| , \epsilon \} \mbox{ with probability } \geq 1-\gamma   
\end{equation*}
\end{linenomath*}
\item  \qquad \qquad Compute $\delta_j^{t,m,l} = f'(z^{t,m,l}_j)s_j^{t,m,l}$ and store 
$\delta_j^{t,m,l}$, 
$-\frac{\eta^{t,l}}{M} \delta_j^{t,m,l}\norm{a^{t,m,l-1}}$, and $-\frac{\eta^{\tau,l}}{M} \norm{\delta^{t,m,l}}a_j^{t,m,l-1}$
in qRAM.
\end{enumerate}
\end{subroutine}

Similar to the feedforward case, the running time of the quantum backpropagation algorithm is 
\begin{linenomath*}
\begin{equation*}
 \tilde{O} \lp  \sqrt{TM}N \frac{\log (1/\gamma)}{\epsilon}R^{t,m}_\delta\rp,   
\end{equation*}
\end{linenomath*}
where $R_\delta^{t,m} = \frac{1}{N-n_l}\sum_{l=1}^{L-1}\sum_{j=1}^{n_l} 
\frac{\norm{\tilde{X}^{[t,l+1,j]}}_F\norm{\delta^{t,m,l+1}} }{\max\left\{1,\abs{\ip{\lp W^{t,l+1}\rp^T_j}{\delta^{t,m,l+1}}} \right\}}$.

%Again, when we say that we store the values $-\frac{\eta^{t,l}}{M} \delta_j^{t,m,l}\norm{a^{t,m,l-1}}$ and $-\frac{\eta^{\tau,l}}{M} \norm{\delta^{t,m,l}}a_j^{t,m,l-1}$ in qRAM, we mean we store these values as leaves of a tree and update the intermediate tree nodes, so that the states  $\ket{X^{[t,l,j]}}$ and $\ket{\tilde{X}^{[t,l,j]}}$ can be constructed in $\polylog(TMN)$ time and Lemma \ref{clm:W_runtime} can be applied.  

The quantum training algorithm consists of running the feedforward and the backpropagation algorithms for all inputs in a mini-batch of size $M$, and iterating this procedure $T$ times. After each mini batch has been processed, we explicitly update the biases $b$ in qRAM, but we do not explicitly update the weights $W$, since this would take time $\tilde{O}(E)$. Instead, we compute an estimate of the norm of the rows and columns of the weight matrices and keep a history of the $a$ and $\delta$ vectors in memory so that we can create the quantum states corresponding to the weights on the fly. 

\begin{subroutine}(\textbf{Quantum }$(\epsilon,\gamma)$-\textbf{Training})\label{sub:q-training}\\

\noindent
{\bf Inputs}: input pairs $(x^{t,m},y^{t,m})$ for all $t \in T, m \in [M]$, parameters $\eta^{t,l}$, for $t \in [T], l \in [L]$, and $\epsilon, \gamma >0$.
\begin{enumerate}
    \item Initialise the weights and biases $W^{1,l}, b^{1,l}$ for $l \in [L]$ with a low-rank initialization.
    \item For $t=1$ to $T$ do:
    \item \qquad For $m=1$ to $M$ do:
    \item \qquad \qquad Run the quantum $(\epsilon,\gamma)$-feedfoward algorithm. 
    \item \qquad \qquad Run the quantum $(\epsilon,\gamma)$-backpropagation algorithm. 
    \item \qquad Compute the biases and update the qRAM with
\begin{linenomath*}
\begin{align*}
b^{t+1,l}_j &= b^{t,l}_j - \eta^{t,l} \frac{1}{M}\sum_{m} \delta^{t,m,l}_j
\end{align*}
\end{linenomath*}
\item \qquad Compute the estimates of the norms $\overline{\norm{W^{t+1,l}_j}}$ and $\overline{\norm{(W^{t+1,l})^j}}$ with relative error $\xi=\epsilon/3$. 
\end{enumerate}
\end{subroutine}

\begin{thm}\label{thm:main_result}
The running time of the quantum training algorithm is 
$\tilde{O} \lp (TM)^{1.5}N  \frac{\log (1/\gamma)}{\epsilon}R \rp$, where $R = R_a+R_\delta+R_W$, $R_a = \frac{1}{TM}\sum_{t,m}R^{t,m}_a$, $R_\delta = \frac{1}{TM}\sum_{t,m}R^{t,m}_\delta$, and 
$R_W = \frac{1}{T}\sum_t (R^t_{W_r}+R^t_{W_c})$, 
with $R^t_{W_r} =  \frac{1}{M} \frac{1}{N-n_1}\sum_{l=2}^L\sum_{j=1}^{n_l} \frac{\norm{X^{[t,l,j]}}_F}{\norm{W^{t,l}_j}}$ and  $R^t_{W_c} = \frac{1}{M} \frac{1}{N-n_1}\sum_{l=2}^L\sum_{j=1}^{n_l} \frac{\norm{\tilde{X}^{[t,l,j]}}_F}{\norm{\lp W^{t,l}\rp^T_j}}$.

\end{thm}

\begin{proof}

The terms $R_a$ and $R_\delta$ come from the feedforward and backpropagation subroutines. The $R_W$ term comes from the estimation of the norms $\norm{W^{t,l}_j}$ and $\norm{\lp W^{t,l}\rp^T_j}$ which only happens once for each mini-batch. For a given $t,l,j$ the estimation of the norm $\overline{\norm{W^{t+1,l}_j}}$ takes time $O(T_W / \xi)$, with $T_W= \tilde{O}\lp \frac{\norm{X^{[t,l,j]}}_{F}}{\norm{W^{t,l}_j}}\sqrt{TM}\rp$, and we take $\xi = \epsilon/3$. 
 Hence, we get the ratio $R^t_{W_r} =  \frac{1}{M} \frac{1}{N-n_1}\sum_{l=2}^L\sum_{j=1}^{n_l} \frac{\norm{X^{[t,l,j]}}_F}{\norm{W^{t,l}_j}}$ and similarly for the estimation of the columns $R^t_{W_c} = \frac{1}{ M} \frac{1}{N-n_1}\sum_{l=2}^L\sum_{j=1}^{n_l} \frac{\norm{\tilde{X}^{[t,l,j]}}_F}{\norm{\lp W^{t,l}\rp^T_j}}$.
Then we can define $R_W = \frac{1}{T}\sum_t (R^t_{W_r}+R^t_{W_c})$ as needed. 

\end{proof}

The trade-off in avoiding an $O\lp E\rp$ scaling is an extra cost factor of $O(\sqrt{TM})$, which makes the overall running time of our quantum training algorithm essentially $\tilde{O}\lp (TM)^{1.5}N\rp$ and, while the exact running time also depends on various other factors, loosely speaking has an advantage over the classical training when $\sqrt{TM} \ll N$.

\subsection{Quantum Evaluation}

Once a neural network has been trained, it can be evaluated on new data to output a predicted label.  While the initial training may only occur once, evaluation of new data labels may occur thousands or millions of times thereafter. There are thus large gains to be had from even small improvements to the efficiency of network evaluation, and we realize such an improvement with our second quantum algorithm.  

The quantum procedure for neural network evaluation is essentially the same as the quantum feedforward algorithm. Assume there is a new input pair $(x,y)$ that we want to evaluate, and that we have unitaries $U_{W_j^{l},a^{l-1}}$ for creating $\ket{W^{l}_j}$ and $\ket{a^{l-1}}$ in time $T_U$, estimates of their norms $\overline{\norm{W^{l}_j}}$ and $\overline{{\norm{a^{l-1}}}}$ to relative error at most $\xi = \epsilon/3$, and vectors $b^{l}$ in qRAM for $l \in [L]$. The running time of the quantum feedforward subroutine implies the following.

\begin{thm}\label{thm2}
There exists a quantum algorithm for evaluating an $(\epsilon,\gamma)$-feedforward neural network in time $\tilde{O}\lp T_U N \frac{\log(1/\gamma)}{\epsilon}R_e\rp$,  where $R_e = \frac{1}{N-n_1}\sum_{l=2}^L\sum_{j=1}^{n_l} 
\frac{\norm{W^{l}_j} \norm{a^{l-1}}}{\max \left\{ 1 , \abs{\ip{ W^{l}_j}{a^{l-1}}}   \right\} }$ and $T_U$ is the time required to prepare any of the states $\ket{W^l_j}$ and $\ket{a^{l-1}}$. For a neural network whose weights are already explictly stored in an $\ell_2$-BST, $T_U$ is polylogarithmic in $E$, and the total running time is $\tilde{O}\lp N\frac{\log(1/\gamma)}{\epsilon}R_e\rp$.
\end{thm}

 In contrast, the classical evaluation algorithm requires running time $O(E)$.  For a network with parameters trained via our quantum training algorithm, we store the network weight matrices $W$ in memory only implicitly, which leads to a time $T_U$ scaling as $O(\sqrt{TM})$.  In this case the overall running time equates to $\tilde{O}\lp \sqrt{TM} N \frac{\log(1/\gamma)}{\epsilon}R_e\rp$.

\subsection{Simulation}\label{s:numerics}

In this section we show that $(\epsilon,\gamma)$-networks can achieve comparable results to standard neural networks, and give numerical evidence that the hard-to-estimate parameters $R_a$, $R_\delta$ and $R_W$ may not be large for problems of practical interest.  We give further arguments to support this in the Discussion section.

We classically simulate the training and evaluation of $(\epsilon,\gamma)$-feedforward neural networks on the MNIST handwritten digits data set consisting of 60,000 training examples and 10,000 test examples.  A network with $L=4$ layers and dimensions $[n_1,n_2,n_3,n_4] =[784,100,30,10]$ was used, with $M = 100, T = 7500, \eta = 0.05$, tanh activation function and mean squared error cost function $C = \frac{1}{2M}\sum_{m} \norm{y^{t,m}-a^{t,m,L}}^2$. For this network size we have $N=924$, $E=81,700$. Our results are summarized in Table \ref{tb:standard-low-rank}.  `Standard' weight initialization refers to drawing the initial weight matrix elements from appropriately normalized Gaussian distributions $W^{l}_{jk}\sim\+N\lp 0,\frac{1}{\sqrt{n_{l-1}}}\rp$ \cite{glorot2010understanding}, and `Low rank' refer to weight initialization as described in Section \ref{ss:weights}, with rank $r=6$. In all cases, Gaussian noise drawn from $N(\epsilon/2,0)$ was added to each inner product evaluated throughout the network. Note that no cost function regularization was used to improve the network performance, and as $\epsilon$ was varied the other network hyper-parameters were not re-optimized. We find that for $\gamma = 0.05$ and various values of $\epsilon$, the network achieves high accuracy whilst incurring only modest contributions to the running time from the quantum-related terms $R_a, R_\delta$ and $R_W$. For the $\epsilon = 0.3$ case we calculate the values of $R_a$, $R_\delta$ and the two components of $R_W$, as a function of the number of gradient update steps $t$.  These results appear in Fig. \ref{fig:all}, where we see that these values quickly stabilize to small constants during the training procedure.

%\red{Shall we also plot $R_E$ or skip it?If so, we need to add a line on the table and the graph.}

\begin{table}
%\centering
\begin{tabular}{c|cccc}
     $\epsilon$ & 0  & 0.1  &  0.3  & 0.5 \\ \hline
     Standard & 96.9\% & 97.1\% & 96.9\% & 96.2\%\\ \hline
     Low rank & - & 96.8\% &  96.8\% &  96.3\% \\
     $\log(1/\gamma)/\epsilon$ & - &   30  & 10  & 6 \\
     $R_a$ &- & 24.4 & 21.9 &  21.3\\
     $R_\delta$ & -& 0.6 &  0.1 &   1.3\\
     $R_W$ &- & 0.1 &  0.1 &  0.1
\end{tabular}
\caption{$(\epsilon,\gamma)$-feedforward neural network accuracy (percentage of correctly labelled test points) on the MNIST handwriting data set, for various values of $\epsilon$, and $\gamma = 0.05$. A network with $L=4$ layers and dimensions $[n_1,n_2,n_3,n_4] =[784,100,30,10]$ was used, with $M = 100, T = 7500, \eta = 0.05$, tanh activation function and mean squared error cost function.}\label{tb:standard-low-rank} 
\end{table}

While the MNIST data set is evaluated using Gaussian distributed noise, we also numerically evaluate the performance of $(\epsilon,\gamma)$-feedforward networks on the Iris flower data set using both Gaussian noise and, for comparison, noise corresponding to the quantum RIPE subroutine where the norms of all vectors are known exactly (see Supplemental Material).  This data set---consisting of 120 training examples and 30 test examples, with each data point corresponding to a length $4$ real vector and having one of three labels---is small enough that classically sampling from the RIPE distribution (which corresponds to simulating the output of a quantum circuit) during the network training can be performed efficiently on a standard desktop computer. A network with $L=3$ layers, dimension $[n_1,n_2,n_3] = [4,10,3]$, tanh activation function and mean squared error cost function was chosen, with $M = 10$, $T=1200$ and $\eta = 0.07$.  For each $(\epsilon,\gamma)$ pair we repeated the network training and evaluation 5 times and show the average number of correctly labelled test points  in Table \ref{tb:iris}. For $\gamma = 0.05$ and the same choices of $\epsilon$ as in the MNIST case, we find that the choice of noise distribution, Gaussian or RIPE, does not lead to significant difference in the network performance.   

\begin{table}
%\centering
\begin{tabular}{c|cccc}
     $\epsilon$ & 0  & 0.1  &  0.3  & 0.5 \\ \hline
     Gaussian & 28.6 & 29.8 & 29.2 & 29.0\\ \hline
     RIPE & - & 28.6 & 28.4 & 28.0 \\
     $\log(1/\gamma)/\epsilon$ & - &   30  & 10  & 6 \\
     $R_a$ &- & 6.0 & 4.4 &  4.0\\
     $R_\delta$ & -& 0.7 &  0.6 &   0.6\\
     $R_W$ &- & 0.4 &  0.4 &  0.4
\end{tabular}
\caption{Average number of correctly labelled test points (out of a total of $30$) for the Iris flower data set, for Gaussian distributed noise as well as noise from the RIPE distribution.  For the case of RIPE, the $R_a, R_\delta$ and $R_W$ terms presented are also averaged over $5$ network evaluations. $(\epsilon,\gamma)$-feedforward neural network accuracy (percentage of correctly labelled test points) were trained on the Iris  data set, for various values of $\epsilon$, and $\gamma = 0.05$. A network with $L=3$ layers and dimensions $[n_1,n_2,n_3] =[4,10,3]$ was used, with $M = 10, T = 1200, \eta = 0.07$, tanh activation function and mean squared error cost function.}\label{tb:iris} 
\end{table}

%\begin{table}
%\centering
%\begin{tabular}{c|cccc}
%     $\epsilon$ & 0  & 0.1  &  0.3  & 0.5 \\ \hline
%     Gaussian & 27.8 & 27.8 & 28.2 & 28.0\\ \hline
%     RIPE & 27.8 & 28.0 &  28.0 &  28.8 
%\end{tabular}
%\caption{Average number of correctly labelled test points (out of a total of $30$) for the Iris flower data set, for Gaussian and RIPE distributed noise.  $(\epsilon,\gamma)$-feedforward neural networks were trained for various values of $\epsilon$, and $\gamma = 0.05$. A network with $L=3$ layers and dimensions $[n_1,n_2,n_3] =[4,10,3]$ was used, with $M = 10, T = 600, \eta = 0.5$, sigmoid activation function and mean squared error cost function. }\label{tb:iris} 
%\end{table}

\subsection{Quantum-inspired classical algorithms}

Our quantum training and evaluation algorithms rely on the $\ell_2$-BST data structure in order to efficiently estimate inner products via the RIPE procedure.  By a result of Tang, such a data structure in fact allows inner products to be efficiently estimated classically as well:

\medskip\textbf{Classical $\ell_2$-BST-based Inner Product Estimation \cite{tang2018quantum}}.  If $x,y\in\mb{R}^n$ are stored in $\ell_2$-BSTs then, with probability at least $1-\gamma$, a value $s$ can be computed satisfying
\begin{linenomath*}
\begin{equation*}
  \abs{s-\< x,y\>} \leq 
  \begin{cases}
   \epsilon  \abs{\< x,y\>} & \text{in time }\widetilde{O}\left(\frac{ \log(1/\gamma)}{ \epsilon^2}\frac{\norm{x}^2\norm{y}^2}{ \abs{\ip{x}{y}}}\right)\\
   \epsilon & \text{in time }\widetilde{O}\left(\frac{ \log(1/\gamma)}{ \epsilon^2} \norm{x}^2\norm{y}^2\right)
  \end{cases}
\end{equation*}
\end{linenomath*}

We can use this concept to derive quantum-inspired classical analogues of our algorithms. If the network weights are explicitly stored in $\ell_2$-BSTs, then the quantum RIPE procedure can be directly replaced by this classical inner product estimation routine to give a classical algorithm for network evaluation which has running time
\begin{linenomath*}
\begin{equation*}
    \tilde{O}\lp N\frac{\log(1/\gamma)}{\epsilon^2}R^{cl}_e\rp, 
\end{equation*}
\end{linenomath*}
where $R^{cl}_e = \frac{1}{N-n_1}\sum_{l=2}^L\sum_{j=1}^{n_l} 
\frac{\norm{W^{l}_j}^2 \norm{a^{l-1}}^2}{\max \left\{ 1 , \abs{\ip{ W^{l}_j}{a^{l-1}}}^2   \right\} }\ge R_e^2$.

Dequantizing the quantum training algorithm is slightly more complicated as the network weights are only stored implicitly.  Nonetheless, the required inner products can still be estimated efficiently, and in the Methods section we show a quantum-inspired algorithm can be given which runs in time
\begin{linenomath*}
\begin{equation*}
\tilde{O} \lp (TM)^{2}N  \frac{\log (1/\gamma)}{\epsilon^2}\lp R_a^{cl} + R_\delta^{cl}\rp \rp    
\end{equation*}
\end{linenomath*}
where $R_{a}^{cl}\ge R_a^2$ and $R_\delta^{cl} \ge R_\delta^2$.  

The quantum-inspired algorithm is thus slower than the quantum one by a factor of $\frac{\sqrt{TM}}{\epsilon}\frac{R_a^2 + R_\delta^2}{ \lp R_a + R_\delta + R_W\rp}$.  We analyze their relative performance further in the Discussion section.

\section{Discussion}\label{sec:discussion}
We have presented quantum algorithms for training and evaluating $(\epsilon,\gamma)$-feedforward neural networks which take time 
$\tilde{O} \lp (TM)^{1.5}N  \frac{\log (1/\gamma)}{\epsilon}R \rp$ and $\tilde{O}(NT_U\frac{\log (1/\gamma)}{\epsilon}R_e)$ respectively.  These are the first algorithms based on the canonical classical feedforward and backpropagation procedures with running times better than the number of inter-neuron connections, opening the way to utilizing larger sized neural networks.  Furthermore, our algorithms can be used as the basis for quantum-inspired classical algorithms which have the same dependence on the network dimensions, but a quadratic penalty in other parameters.  Let us now make a number of remarks on our design choices and the performance of our algorithms.

\medskip\textbf{Performance of $(\epsilon,\gamma)$-feedforward neural networks.} While the notion of introducing noise in the inner product calculations of an $(\epsilon,\gamma)$-feedforward neural network is, as previously mentioned, similar in spirit to known classical techniques for improving neural network performance, there are certain differences. Classically, dropout and multiplicative Gaussian noise \cite{srivastava2014dropout} involve introducing perturbations post-activation: $a^l_j \ra a^l_j r_j$, where $r_j\sim N(1,1)$ (multiplicative Guassian noise) or $r_j\sim Bernoulli(p)$ (dropout), whereas in our case the noise due to inner product estimation occurs pre-activation. Furthermore, classical methods are typically employed during the training phase but, once the network parameters have been trained, new points are evaluated without the introduction of noise.  This is in contrast to our quantum algorithm where the evaluation of new points inherently also involves errors in inner product estimation.  However, numerical simulation (see Tables \ref{tb:standard-low-rank} and \ref{tb:iris}) of the noise model present in our quantum algorithm shows that $(\epsilon,\gamma)$-networks may tolerate modest values of noise, and values of $\epsilon$ and $\gamma$ can be chosen for which the network performance does not suffer significantly, whilst at the same time do not contribute greatly to the factor of $\log(1/\gamma)/\epsilon$ that appears in the quantum running time.

\medskip\textbf{Low rank initialization.} One assumption we make in our quantum algorithm is a low-rank initialization of the network weight matrices, compared with freedom to choose full-rank weights classically. This assumption is made in order to avoid a time of $\tilde{O}(E)$ to input the initial weight values into qRAM.  Low rank approximations to network weights have found applications classically in both speeding up testing of trained networks \cite{sainath2013low, denton2014exploiting, yu2017compressing} as well as in network training \cite{tai2015convolutional}, in some cases delivering significant speedups without sacrificing much accuracy. We find numerically that the low-rank initialization we require for our quantum algorithm works as well as full rank for a range of $\epsilon$ values (see Table \ref{tb:standard-low-rank}). 

\medskip\textbf{Quantum training running time.} Compared with the classical running time of $O(TME)$, our quantum algorithm scales with the number of neurons in the network as opposed to the number of edges.  However, this comes at the cost of a square root penalty in the number of iterations and mini-batch size, and it is an open question to see how to remove this term. The quantum algorithm also has additional factors of $\log(1/\gamma)/\epsilon$ and $R = R_a + R_\delta + R_W$ which do not feature in the classical algorithm.  While the $\epsilon$ and $\gamma$ can be viewed as hyperparameters, which can be freely chosen, the impact of the $R$ terms warrant further discussion.  

The ratio $\norm{X^{[t,l,j]}}/\norm{W^{t,l}_j}$ appears in both $R_a$ and $R_W$ (in the $R^t_{W_r}$ contribution) and 
similarly the ratio $\norm{\tilde{X}^{[t,l,j]}}/\norm{\lp W^{t,l}\rp^T_j}$ appears in $R_\delta$ and the $R^t_{W_r}$ contribution to $R_W$. 
While exact values may be difficult to predict, one can expect the following large $t$ behaviour: Classically, initial weight matrices are typically chosen so that the entries are drawn from a normal distribution with standard deviation $1/\sqrt{\text{row length}}$, which would give $\norm{W^{1,l}_j}\approx 1$, and a similar scenario can hold with low rank initialization.  As the weights are updated according to equation \eqref{eq:w_update}, and since changes in individual matrix elements may be positive or negative, for a constant step size $\eta$ one expects $\norm{W^{t,l}_j}$ to roughly grow proportionally to $\sqrt{t\eta}$. The norm $\norm{X^{[t,l,j]}}$ has value 
\begin{linenomath*}
\begin{equation*}
 \norm{X^{[t,l,j]}}_{F} = \sqrt{\sum_{\tau=0}^{t-1}\sum_{\mu=1}^M \lp \frac{\eta}{M} \delta_j^{\tau,\mu,l}\norm{a^{\tau,\mu,l-1}}\rp^2}   
\end{equation*}
\end{linenomath*}
 and, as $t$ increases and the network becomes close to well trained, we expect $\delta^{\tau,\mu,l}_j \rightarrow 0$, whereas for activations bounded in the range $[-1,1]$, as is the case for the tanh function,  $\norm{a^{t,l-1}} \le \sqrt{n_{l-1}}$. We thus expect $\norm{X^{[t,l,j]}}$ to saturate for large $t$ and not grow in an unbounded fashion. Fig. \ref{fig:all}(a) showing the time averaged values of $R^{t}_{W_r}$ and $R^t_{W_c}$ for an $(\epsilon,\gamma)$-network trained on the MNIST data set is consistent with these ratios saturating to very small values over time, in fact values less than $ 0.014$. A similar result can be seen in Fig. \ref{fig:iris-all}(a) for the IRIS data set.
 
 %More conservatively, if one assumes that instead of tending to zero, $\delta^{\tau,\mu,l}_j$ remains roughly constant over time, one would expect
 %$$\norm{X^{[t,l,j]}} = O \lp \eta \sqrt{t n_{l-1}/M}\rp$$
%In this case $\norm{X^{[t,l,j]}}/\norm{X^{t,l}_j}$ should, for large $t$, be roughly proportional to $\sqrt{\eta n_{l-1}/M}$. 

The term $R_a = \frac{1}{TM}\sum_{t,m}R_a^{t,m}$ is an average over iterations and mini-batch elements of terms $R_a^{t,m}$, which themselves are averages over neurons in the network of ratios and products of matrix and vector norms:
\begin{linenomath*}
\begin{align*}
    R^{t,m}_a 
&= \frac{1}{N-n_1}\sum_{l=2}^L\sum_{j=1}^{n_l} 
\frac{ \norm{X^{[t,l,j]}} \norm{a^{t,m,l-1}}    } { \max\left\{1, \abs{\ip{W^{t,l}_j}{a^{t,m,l-1}}}\right\}}
\end{align*}
\end{linenomath*}
As discussed, we expect $\norm{X^{[t,l,j]}}$ to saturate for large $t$, and for activations in $[-1,1]$ we have $\norm{a^{t,m,l-1}} \le \sqrt{n_{l-1}}$.  However, as the network becomes well trained, one expects the inner products $\ip{W^{t,l}_j}{a^{t,m,l-1}}$ to become large in magnitude so that the neurons have post-activation values close to $\pm 1$.  It is thus reasonable to expect the $R_a$ to saturate or even decline for large $t$.  This is consistent with our results in Fig. \ref{fig:all}(b) and \ref{fig:iris-all}(b).  One expects similar large $t$ behaviour for $R_\delta$, except intuitively $R_\delta$ should be much smaller than $R_a$ since $\norm{\delta^{t,m,l}}$ should become very small as the network becomes well trained.  Figs. \ref{fig:all}(c) and \ref{fig:iris-all}(c) display this expected behaviour. While these simulation results are already promising, we expect the quantum advantage in the running time to become more prominent for larger size neural networks.

\medskip\textbf{Quantum evaluation running time.} The classical running time of the evaluation algorithm is $O(E)$.  In contrast, our quantum evaluation algorithm runs in time 
$\tilde{O}\lp N  \frac{\log(1/\gamma)}{\epsilon}R_e\rp$ if the entries of the weight matrices are explicitly stored in qRAM.  By the same arguments given above for the training algorithm, we expect the penalty term $R_e$ to be small for well-trained networks, and in this case we expect the quantum training algorithm to be able to provide a significant speed-up over its classical counterpart. 

\medskip\textbf{Quantum vs. quantum-inspired running times.}
While both the quantum and quantum-inspired classical algorithms presented here have a running time linear in $N$, the quantum-inspired algorithms come with a quadratic overhead in other parameters.  A comparison of their running times is given in Table \ref{tb:quantum-inspired}.   The standard classical training time of $TME$ is outperformed by our  quantum training algorithm when $\sqrt{TM}\frac{N}{E}\frac{\log(1/\gamma)}{\epsilon }\lp R_a + R_\delta + R_W\rp \ll 1$, whereas the quantum-inspired classical training algorithms can only outperform the standard classical algorithm when $TM\frac{N}{E}\frac{\log(1/\gamma)}{\epsilon^2 }\lp R_a^{cl} + R_\delta^{cl}\rp \ll 1$. To understand the significance of the quadratic overheads required by the quantum-inspired algorithms, consider a back-of-the-envelope calculation based on the landmark neural network of Krizhevsky, Sutskever and Hinton \cite{krizhevsky2012imagenet}.  The convolutional neural network they use to recognize the ImageNet LSVRC data set has $TM \approx 10^8$, and fully connected final layers corresponding to $N \approx 2\times 10^4, E \approx 6\times 10^7$.  Taking $\epsilon =0.1, \gamma = 0.05$ and assuming values of $R_a, R_\delta$ and $R_W$ the same order of magnitude as we numerically evaluated for the MNIST case gives $\sqrt{TM}\frac{N}{E}\frac{\log(1/\gamma)}{\epsilon }\lp R_a + R_\delta + R_W\rp \approx 10^3$ and  $TM\frac{N}{E}\frac{\log(1/\gamma)}{\epsilon^2 }\lp R_a^{cl} + R_\delta^{cl}\rp \approx 1.4\times 10^9$, a factor of $10^6$ in favor of the quantum algorithm. While neither the quantum nor the quantum-inspired classical algorithms can compete with the standard classical algorithm in this particular case, plausible changes to the network and data set parameters can be chosen where the quantum training algorithm has an advantage. The quantum-inspired algorithm though would require network parameters many orders of magnitude different to the practical example considered here. In fact, it remains to be seen if there are any real cases where quantum-inspired algorithms can be better in practice than the standard classical ones, and the evidence so far is negative \cite{arrazola2019quantum}.

\begin{table}
    \begin{tabular}{c|c|c}
         & Training & Evaluation  \\ \hline
        Quantum &$\tilde{O} \lp (TM)^{1.5}N  \frac{\log (1/\gamma)}{\epsilon}\lp R_a+R_\delta+R_W\rp \rp$ & $\tilde{O}\lp N\frac{\log(1/\gamma)}{\epsilon}R_e\rp$\\
        Quantum-inspired & $\tilde{O} \lp (TM)^{2}N  \frac{\log (1/\gamma)}{\epsilon^2}\lp R_a^{cl} + R_\delta^{cl}\rp \rp$ & $\tilde{O}\lp N\frac{\log(1/\gamma)}{\epsilon^2}R^{cl}_e\rp$ \\
        Standard classical & $O(TME)$ & $O(E)$
    \end{tabular}
    \caption{Comparison of running times between the quantum, quantum-inspired and standard classical algorithms.  Note that $R_a^{cl}\ge R_a^2$, $R_\delta^{cl}\ge R_\delta^2$ and $R_e^{cl}\ge R_e^2$.}\label{tb:quantum-inspired}
\end{table}

\medskip
Let us add a final remark. In our training algorithm we used classical inputs and showed that the number of iterations required for convergence is similar to the case of classical robust training. One can also consider using superpositions of classical inputs for the training, which could conceivably  reduce the number of iterations or size of mini-batch required. We leave this as an interesting open direction for future work. 

%\red{add sth on evaluation?}

\begin{figure}%
\includegraphics[height=4.0cm]{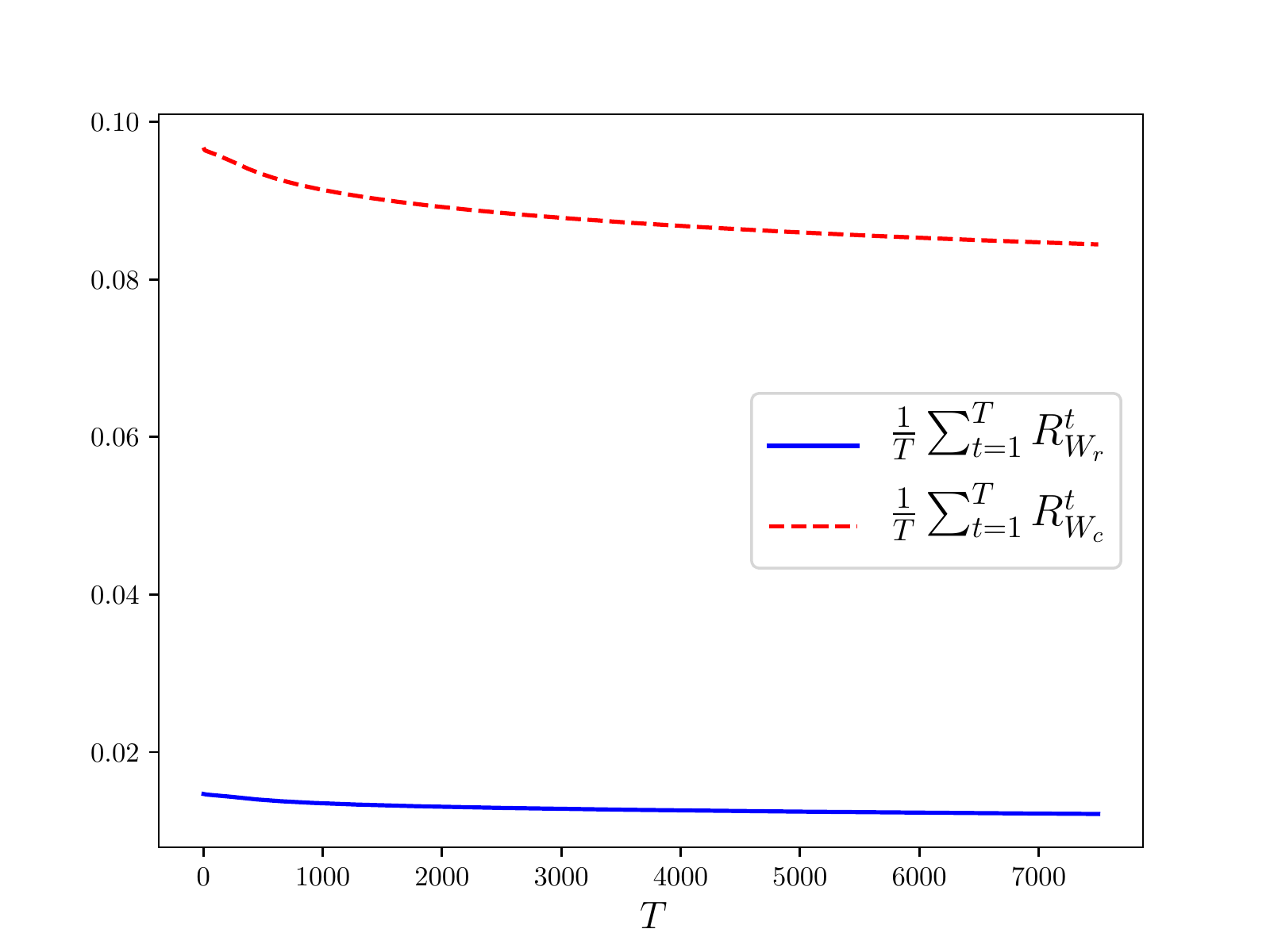}
\includegraphics[height=4.0cm]{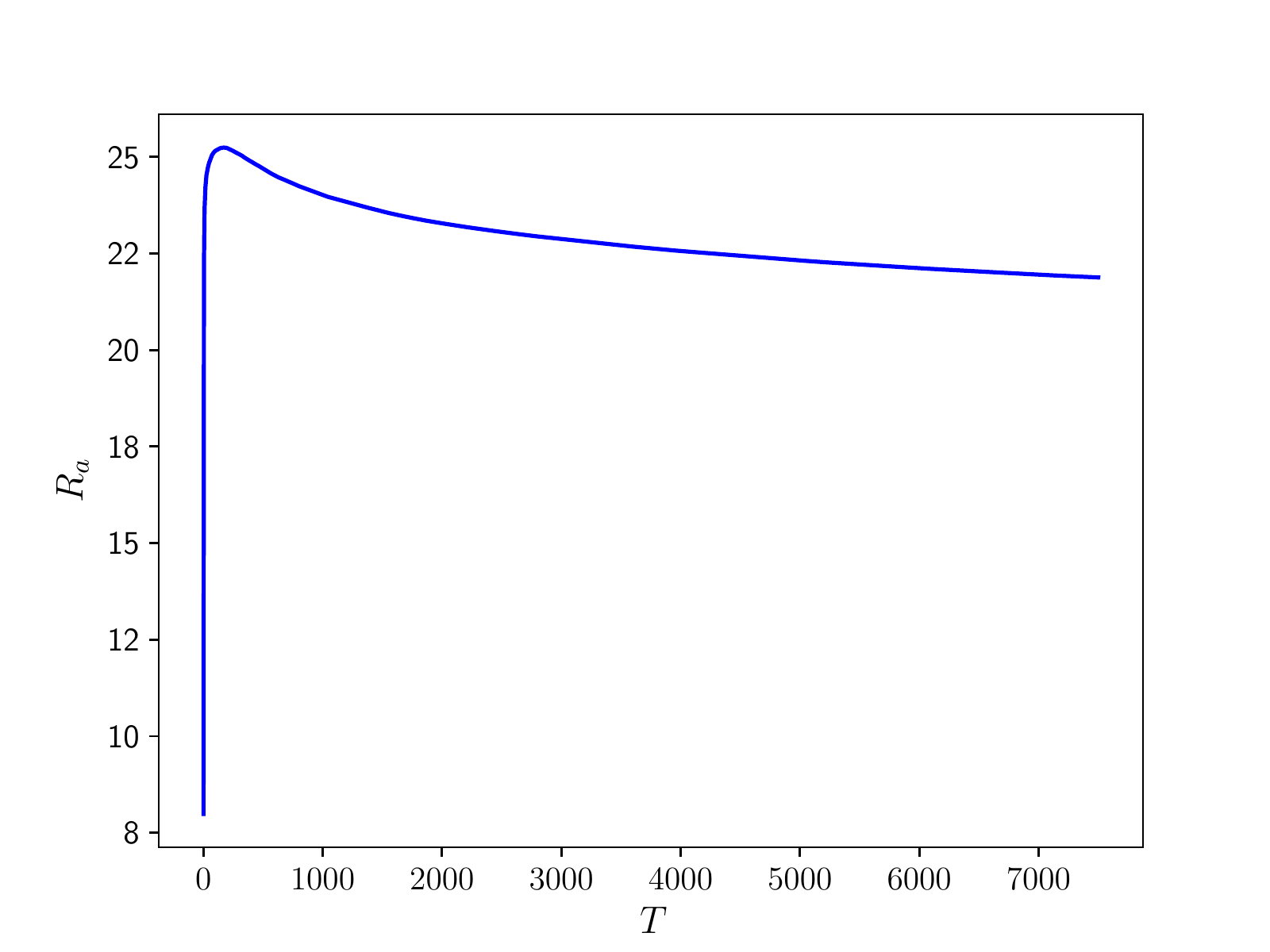}
\includegraphics[height=4.0cm]{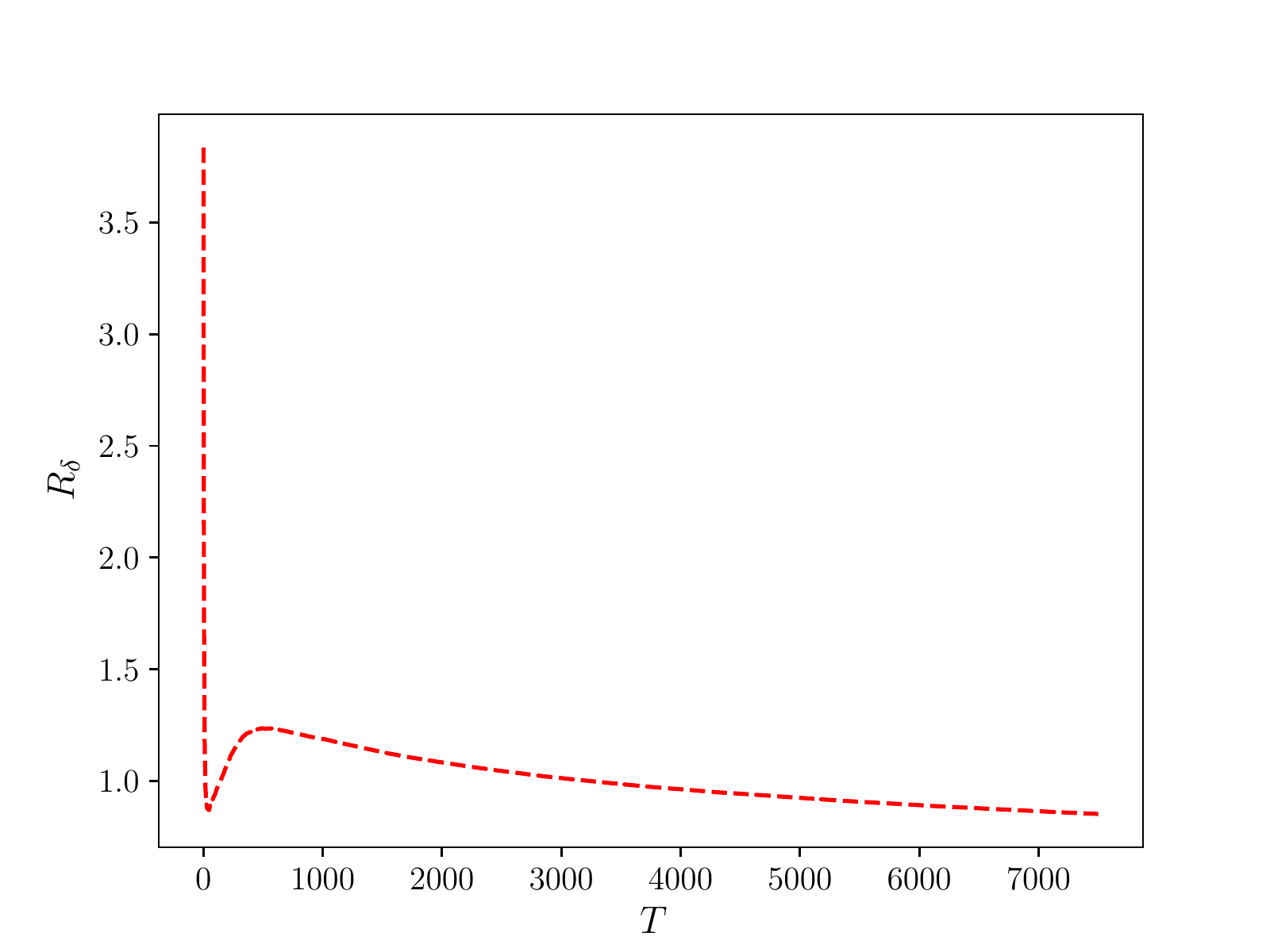}
\caption{Evolution of the terms $R_{W}$, $R_a$, and $R_\delta$ as a function of the number of iterations $T$, for the $(\epsilon,\gamma) = (0.3,0.05)$ network of Table \ref{tb:standard-low-rank} (MNIST data set). }
\label{fig:all}%
\end{figure}

\begin{figure}%
\includegraphics[height=4.0cm]{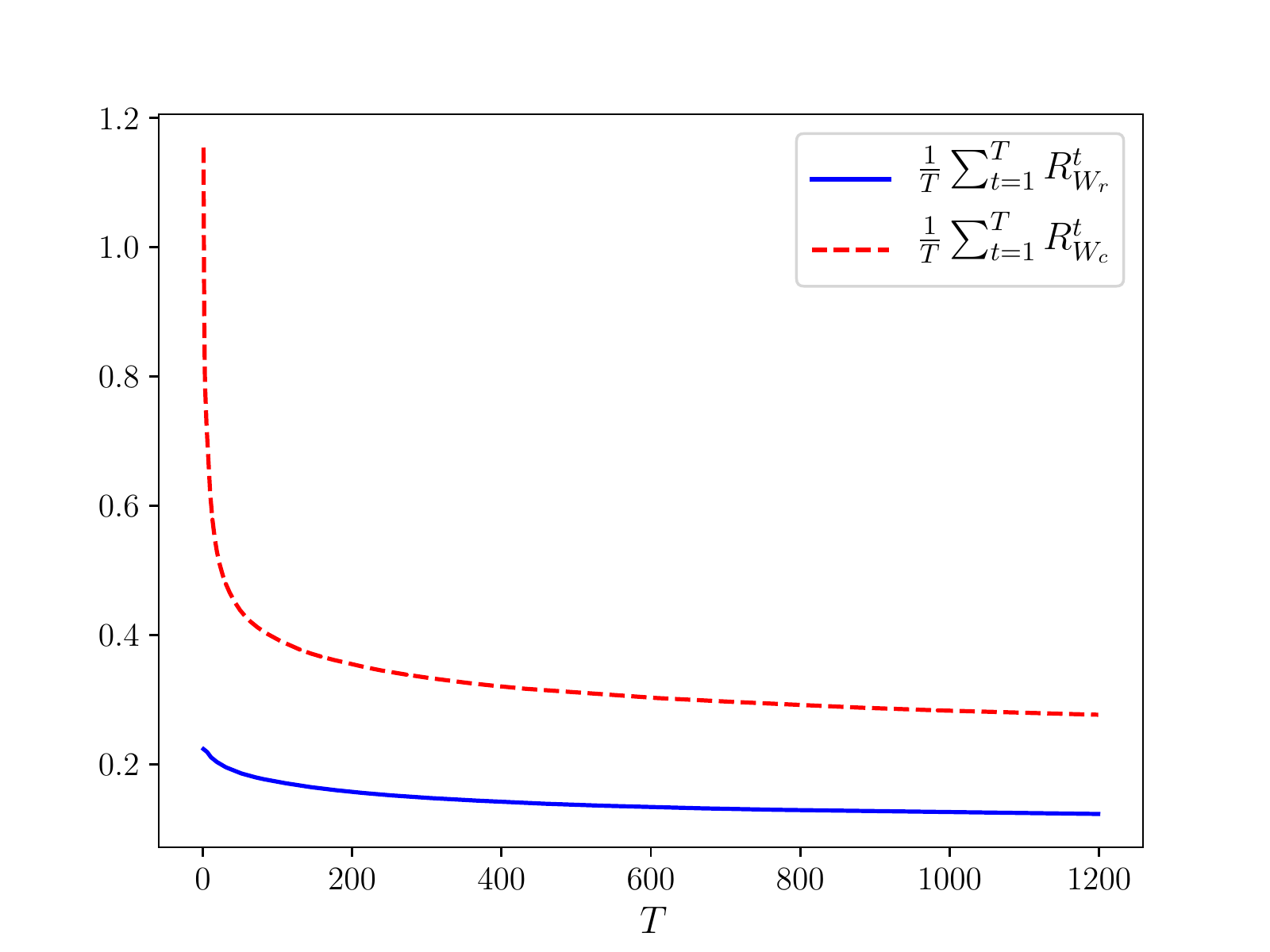}
\includegraphics[height=4.0cm]{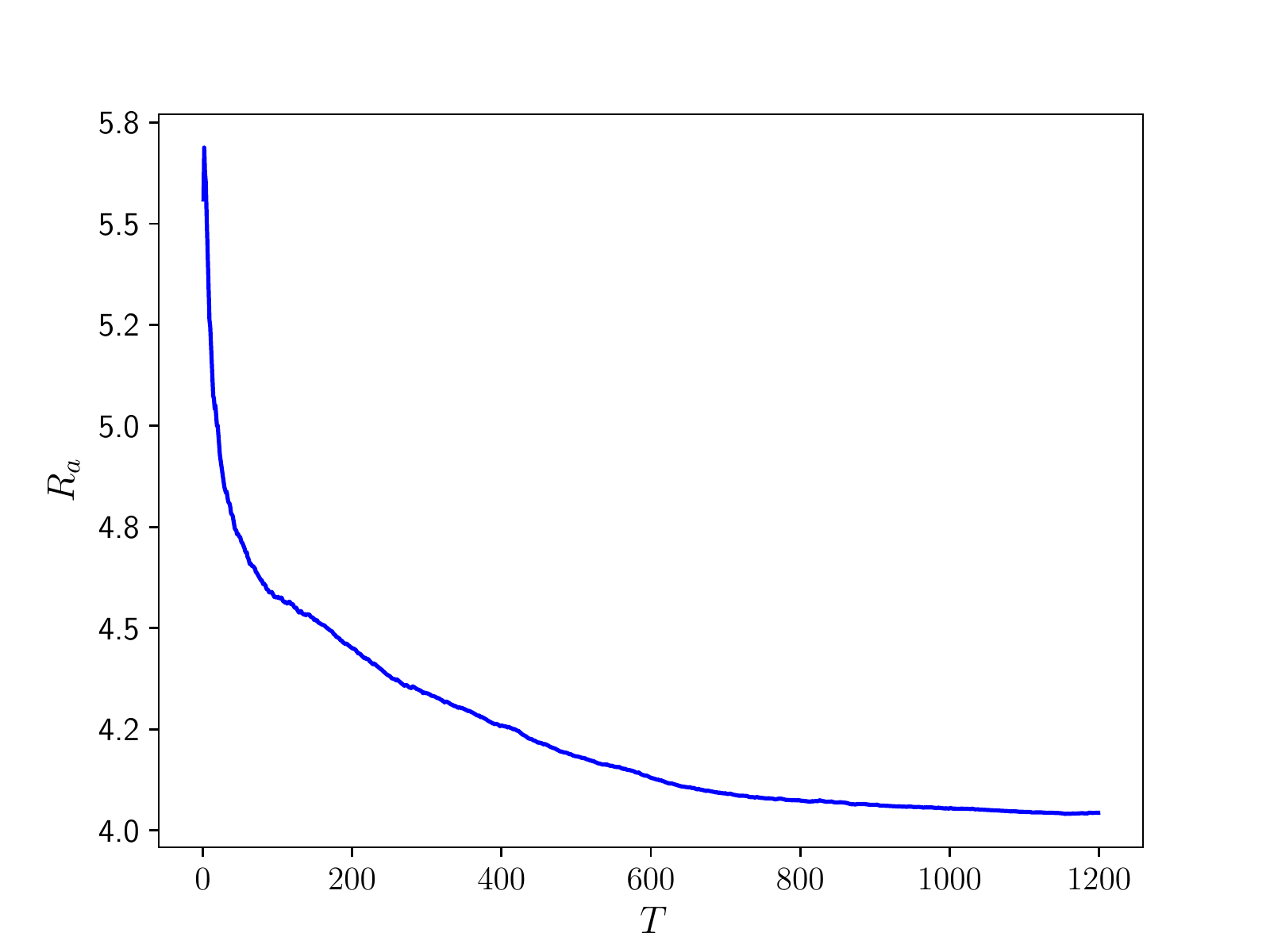}
\includegraphics[height=4.0cm]{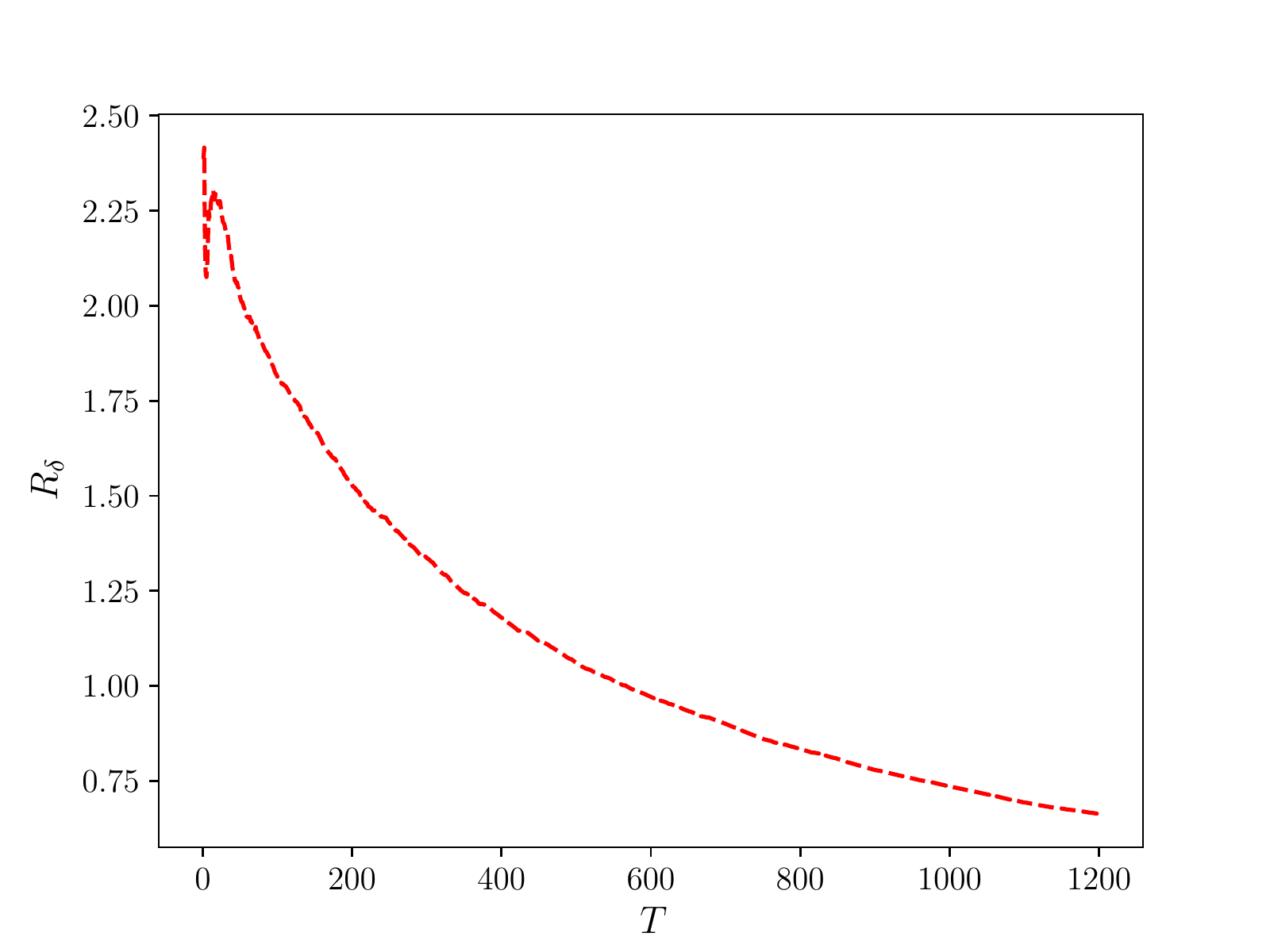}
\caption{Evolution of the terms $R_{W}$, $R_a$, and $R_\delta$ as a function of the number of iterations $T$, for one typical instance of an $(\epsilon,\gamma) = (0.3,0.05)$ network of Table \ref{tb:iris} (Iris data set). }
\label{fig:iris-all}%
\end{figure}

\section{Methods}

\subsection{Robust Inner Product Estimation running time.} 

In \cite{kerenidis2018qmeans}, the authors give a quantum inner product estimation (IPE) algorithm which allows one to perform the mapping $\ket{x}\ket{y}\ket{0}\ra \ket{x}\ket{y}\ket{s}$, where $s$ satisfies $\abs{s-\< x,y\>} \leq  \epsilon$ with probability at least $1-\gamma$.  The running time is $\widetilde{O}\left(\frac{T_U  \log(1/\gamma)}{ \epsilon}\norm{x}\norm{y}\right)$, where $T_U$ is the time require to implement unitary operations for creating states $\ket{x}, \ket{y},\ket{\norm{x}},\ket{\norm{y}}$.  

The RIPE algorithm is a generalization to the case where one only has access to estimates $\overline{\norm{x}},\overline{\norm{y}}$ of the vector norms, satisfying  $\abs{\overline{\norm{x}} - \norm{x}}\leq \frac{\epsilon}{3} \norm{x}$ and $\abs{\overline{\norm{y}} - \norm{y}}\leq \frac{\epsilon}{3} \norm{y}$ and where one would like to obtain inner product estimates to either additive or multiplicative error. For normalized vectors $\ket{x}, \ket{y}$, the IPE algorithm runs in time $\widetilde{O}\lp \frac{T_U\log(1/\gamma)}{\epsilon'}\rp $ and outputs $\abs{s-\braket{x}{y} } \leq  \epsilon'$.
By taking $ \epsilon' = \frac{\epsilon}{4} \abs{\braket{x}{y}}$ we obtain a relative error algorithm in time
$\widetilde{O}\left(\frac{T_U \log(1/\gamma)}{ \epsilon}\frac{1}{ \abs{\braket{x}{y}}}\right)$. Outputting the estimator 
$s'=\overline{\norm{x}}\,\overline{\norm{y}} s$ then satisfies
\begin{linenomath*}
\begin{eqnarray}
\abs{s' - \<x,y\> } & \leq &  
\abs{s' -  \norm{x}\norm{y} s }+
\abs{ \norm{x}\norm{y}s - \norm{x}\norm{y} \ip{x}{y} }\\
& \leq & [(1+\epsilon/3)^2 - 1 ] \norm{x}\norm{y} (1+\epsilon/4)\<x|y\> + \epsilon/4 \<x,y\> \\
& \leq & \epsilon \abs{\ip{x}{y}}
\end{eqnarray}
\end{linenomath*}
for small enough $\epsilon$. An absolute error estimate can similarly be obtained by taking replacing $\epsilon$ above with $\epsilon/|\< x,y\>|$. We thus have an algorithm that runs in time 
$T_{\text{IPE}}(x,y) = \widetilde{O}\left(\frac{T_U \log(1/\gamma)}{ \epsilon} \frac{\norm{x}\norm{y}}{\max\{1, \abs{\<x,y\>}\}} \right)$ and achieves an error of $|s-\< x,y\>| \leq  \max \left\{ \epsilon \abs{\ip{x}{y}} , \epsilon \right\}$.

\subsection{Constructing the weight matrix states and estimating their norms}\label{ss:weights}

If, at each iteration $t$, the values $a^{t,m,l}_j$ and $-\frac{\eta^{t,l}}{M} \delta_j^{t,m,l}\norm{a^{t,m,l-1}}$ are stored in qRAM for all $m\in[M]$, then the states $\ket{a^{t,m,l-1}}$ and 
\begin{linenomath*}
\begin{equation*}
 \ket{X^{[t,l,j]}} = \frac{1}{\norm{X^{[t,l,j]}}_{F}}\sum_{\tau=0}^{t-1}\sum_{\mu=1}^M -\frac{\eta^{\tau,l}}{M} \delta_j^{\tau,\mu,l}\norm{a^{\tau,\mu,l-1}}\ket{\tau}\ket{\mu}   
\end{equation*}
\end{linenomath*}
can be created coherently in time polylogarithmic in $TMN$. That is, unitary operators $U_X$ and $U_a$ can be implemented in this time that effect the transformations:
\begin{linenomath*}
\begin{align*}
    U_X\ket{t}\ket{l}\ket{j}\ket{0}\ket{0} &\ra \ket{t}\ket{l}\ket{j}\ket{X^{[t,l,j]}} \\
    U_a\ket{l}\ket{t}\ket{m}\ket{0}&\ra 
    \ket{l}\ket{t}\ket{m}\ket{a^{t,m,l-1}}%, \;\;\; \mbox{for} \; \tau=\{0,\ldots,t-1\}, \mu \in [M]
\end{align*}
\end{linenomath*}
Application of $U_X$ on the first five registers of state $\ket{t}\ket{l}\ket{j}\ket{0}\ket{0}\ket{0}$, followed by application of $U_a$ on registers $\{2,4,5,6\}$ produces the state
\begin{linenomath*}
\begin{equation*}
    \ket{t}\ket{l}\ket{j}\frac{1}{\norm{X^{[t,l,j]}}_{F}}\sum_{\tau=0}^{t-1}\sum_{\mu=1}^M-\frac{\eta^{\tau,l}}{M}\delta_j^{\tau,\mu,l}\ket{\tau}\ket{\mu}\sum_k a^{\tau,\mu,l-1}_k\ket{k}.
\end{equation*}
\end{linenomath*}
Applying the Hadamard transformations $\ket{\tau}\ra \frac{1}{\sqrt{t}}\sum_{x}(-1)^{\tau\cdot x}\ket{x}$ and $\ket{\mu}\ra \frac{1}{\sqrt{M}}\sum_{y}(-1)^{\mu\cdot y}\ket{y}$ leads to
\begin{linenomath*}
\begin{align}
&\ket{t}\ket{l}\ket{j}\frac{1}{\norm{X^{[t,l,j]}}_{F}\sqrt{Mt}}\sum_{x}\sum_y\lp \sum_{k=1}^N \sum_{\tau=0}^{t-1}\sum_{\mu=1}^M -\frac{\eta^{\tau,l}}{M}(-1)^{\tau\cdot x +\mu\cdot y}\delta_j^{\tau,\mu,l} a_k^{\tau, \mu, l-1}\ket{k}\rp\ket{x}\ket{y}  \\
&=\ket{t}\ket{l}\ket{j}\lp \sin\theta \ket{W^{t,l}_j}\ket{0}\ket{0} + \cos\theta\ket{\textsf{junk}}\ket{00^\perp}\rp, \label{eq:W-amp-est}
\end{align}
\end{linenomath*}
where
\begin{linenomath*}
\begin{equation*}
 \ket{W^{t,l}_j} = \frac{1}{\norm{W^{t,l}_j}}\sum_{k=1}^{n_l} W^{t,l}_{jk}\ket{k} = \frac{1}{\norm{W^{t,l}_j}}\sum_{k=1}^N \lp \sum_{\tau=0}^{t-1}\sum_{\mu=1}^M \frac{-\eta^{\tau,l}}{M} \delta^{\tau,\mu,l}_j a^{\tau,\mu,l-1}_k\rp \ket{k},   
\end{equation*}
\end{linenomath*}
$\sin^2\theta = \frac{\norm{W^{t,l}_j}^2}{\norm{X^{[t,l,j]}}_{F}^2Mt}$ and $\ket{00^\perp}$ is a state orthogonal to $\ket{0}\ket{0}$.
By the well-known quantum procedures of amplitude amplification and amplitude estimation \cite{brassard2002quantum}, given access to a unitary operator $U$ acting on $k$ qubits such that $U\ket{0}^{\otimes k} = \sin(\theta) \ket{x, 0} + \cos(\theta) \ket{G, 1}$ (where $\ket{G}$ is arbitrary),  $\sin^{2}(\theta)$ can be estimated to additive error $\epsilon \sin^{2}(\theta)$ in time $O\lp \frac{T(U)}{\epsilon \sin(\theta)}\rp$ and 
$\ket{x}$ can be generated in expected time $O\lp \frac{T(U)}{\sin (\theta)}\rp$,where $T(U)$ is the time required to implement $U$.  Amplitude amplification applied to the unitary preparing the state in \eqref{eq:W-amp-est} allows one to generate $\ket{W^{t,l}_j}$ in time 
\begin{linenomath*}
\begin{equation*}
    T_W = O\lp \frac{\norm{X^{[t,l,j]}}_{F}}{\norm{W^{t,l}_j}}\sqrt{TM}\polylog(TMN)\rp.
\end{equation*}
\end{linenomath*}
Similarly, amplitude estimation can be used to find an $s$ satisfying $\abs{s - \frac{\norm{W^{t,l}_j}^2}{\norm{X^{[t,l,j]}}_{F}^2Mt}}\le \xi\frac{\norm{W^{t,l}_j}^2}{\norm{X^{[t,l,j]}}_{F}^2Mt}$ 
in time $ O \lp T_W/\xi\rp$. Outputting  $\overline{\norm{W_j^{t,l}}} = \norm{X^{[t,l,j]}}_{F}\sqrt{Mts}$ then satisfies $\abs{\overline{\norm{W^{t,l}_j}}- \norm{W^{t,l}_j}}\le \xi \norm{W^{t,l}_j}$, since
$| \sqrt{s} \norm{X^{[t,l,j]}}_F\sqrt{Mt} - \norm{W^{t,l}_j} | \leq \xi \norm{W^{t,l}_j} $. 

If the values $\delta^{t,m,l}_j$ and $\frac{-\eta^{t,l}}{M}a^{t,m,l-1}_j\norm{\delta^{t,m,l}}$ are also stored at every iteration,  analogous results also hold for creating quantum states $\ket{\lp W^{t,l}\rp^T_j} $ corresponding to the columns of  $W^{t,l}$, except in this case the key ratio $\norm{W^{t,l}_j}/\norm{X^{[t,l,j]}}_F$ that appeared in the previous proof is replaced by  $\norm{\lp W^{t,l}\rp^j}/\norm{\tilde{X}^{[t,l,j]}}_F$, where $\norm{\tilde{X}^{[t,l,j]}}_F$ is the norm of the quantum state
\begin{linenomath*}
\begin{equation*}
    \ket{\tilde{X}^{[t,l,j]}} = \frac{1}{\norm{\tilde{X}^{[t,l,j]}}_{F}}\sum_{\tau=0}^{t-1}\sum_{\mu=1}^M -\frac{\eta^{\tau,l}}{M} a_j^{\tau,\mu,l-1}\norm{\delta^{\tau,\mu,l}}\ket{\tau}\ket{\mu}.
\end{equation*}
\end{linenomath*}

%Note that at every iteration we can store all the vectors $a^{t,m,l}$ and $\delta^{t,m,l}$ in qRAM in time $O(N\polylog(TMN))$ and later we can create the quantum states $\ket{a^{t,m,l}}$ and $\ket{\delta^{t,m,l}}$ corresponding to these vectors in time $\polylog(TMN)$. 

\subsection{Quantum-inspired classical algorithm for $(\epsilon,\gamma)$-feedforward network training.}

Suppose that for a given $l\in\{2,\ldots, L\}$ and $j\in[n_l]$ the following are stored in an $\ell_2$-BST: 
\begin{linenomath*}
\begin{align*}
  a^{\tau,\mu,l-1} &\qquad \forall \tau\in\{0,1\ldots, t-1\}, \forall \mu\in[M] \\
  a^{t,m, l-1} &\\
  X^{[t,l,j]} &
\end{align*}
\end{linenomath*}
It is then possible to sample $(\tau,\mu)$ with probability $P(\tau,\mu) = \frac{\lp X^{[t,l,j]}\rp^2_{\tau,\mu}}{\norm{X^{[t,l,j]}}^2_F}$, and $k$ with conditional probability $P(k\mid \tau,\mu) = \lp\frac{a^{\tau,\mu,l-1}_k}{\norm{a^{\tau,\mu,l-1}}}\rp^2$ in time $O(\polylog (TMN))$.  The random variable $Z(k,\tau,\mu) \defeq a_k^{t,m,l-1}\frac{\norm{a^{\tau,\mu,l-1}}}{a_k^{\tau,\mu,l-1}} \frac{\norm{X^{[t,l,j]}}_F^2}{\lp X^{[t,l,j]}\rp_{\tau,\mu}}$ can be computed in time $O(\polylog(TMN))$, and has expectation 
\begin{linenomath*}
\begin{align*}
    \langle Z\rangle &=\sum_{k,\tau,\mu}P(k\mid \tau,\mu)P(\tau,\mu)Z(k,\tau,m)\\
&=\sum_{k,\tau,\mu}\frac{a^{\tau,\mu,l-1}_k}{\norm{a^{\tau,\mu,l-1}}} \lp X^{[t,l,j]}\rp_{\tau,\mu}a_k^{t,m,l-1}\\
&=\sum_k \lp \sum_{\tau,\mu}-\frac{\eta^{\tau,l}}{M} \delta_j^{\tau,\mu,l}a_k^{\tau,\mu,l-1}\rp a_k^{t,m,l-1} \\
&= \langle W^{t,l}_j , a^{t,m,l-1}\rangle
\end{align*}
\end{linenomath*}
and variance
\begin{linenomath*}
\begin{align*}
    \sigma^2 &\le \sum_{k,\tau,\mu} \lp a_k^{t,m,l-1}\rp^2 \norm{X^{[t,l,j]}}_F^2 =MT \norm{a^{t,m,l-1}}^2\norm{X^{[t,l,j]}}_F^2. 
\end{align*}
\end{linenomath*}
By the Chebyshev and Chernoff-Hoeffding inequalities, taking the median of $O\lp \log(1/\gamma)\rp$ averages, each an average of $O\lp \frac{1}{\epsilon'^2}\rp$ independent copies of $Z$, produces an estimate $s$ within  $\epsilon' \sigma $ of $\langle Z\rangle$. Taking $\epsilon' = \frac{\epsilon\max\left\{\abs{\langle W^{t,l}_j, a^{t,m,l-1}\rangle},1\right\}}{\sqrt{MT}\norm{a^{t,m,l-1}}\norm{X^{[t,l,j]}}_F}$ then allows an $s^{t,m,l}_j$ to be computed satisfying $\abs{s^{t,m,l}_j - \ip{W^{t,l}_j}{a^{t,m,l-1}}} \le \epsilon \max\left\{\abs{\ip{W^{t,l}_j}{a^{t,m,l-1}}},1\right\}$ with probability at least $1-\gamma$.

Replacing the RIPE procedure in Subroutine \ref{sub:q-ff} with this method for computing $s^{t,m,l}_j$ gives a quantum-inspired classical $(\epsilon,\gamma)$-feedforward subroutine which runs in time 
\begin{linenomath*}
\begin{equation*}
    \tilde{O}\lp TMN\frac{\log(1/\gamma)}{\epsilon^2}{R_a^{cl}}^{t,m}\rp,
\end{equation*}
\end{linenomath*}
where ${R^{cl}_{a}}^{t,m} = \frac{1}{N-n_1}\sum_{l=2}^L\sum_{j=1}^{n_l} \left(\frac{\norm{X^{[t,l,j]}}_F\norm{a^{t,m,l-1}}}{\max\left\{1,\abs{\langle W^{t,l}_j,a^{t,ml-1}\rangle}\right\}}\right)^2$.  Using $\lp \sum_{i=1}^n x_i\rp^2 \le n \sum_{i=1}^n x_i^2$ for $x_i \ge 0$, it follows that ${R^{cl}_{a}}^{t,m}\ge \lp R_a^{t,m}\rp^2$.

The RIPE procedure in Subroutine \ref{sub:q-bp} can similarly be replaced to obtain a quantum-inspired classical $(\epsilon,\gamma)$-backpropagation subroutine which runs in time
\begin{linenomath*}
\begin{equation*}
   \tilde{O}\lp TMN\frac{\log(1/\gamma)}{\epsilon^2}{R_\delta^{cl}}^{t,m}\rp, 
\end{equation*}
\end{linenomath*}
where ${R_\delta^{cl}}^{t,m}\ge \lp R_\delta^{t,m}\rp^2$.

Finally, by substituting the quantum $(\epsilon,\gamma)$-feedforward and quantum $(\epsilon,\gamma)$-backpropagation subroutines with their quantum-inspired classical counterparts in Subroutine \ref{sub:q-training}, one obtains a quantum-inspired classical $(\epsilon,\gamma)$-training algorithm which runs in time 
\begin{linenomath*}
\begin{equation*}
  \tilde{O}\lp \lp TM\rp^2 N\frac{\log(1/\gamma)}{\epsilon^2}\lp R_a^{cl} + R_\delta^{cl}\rp\rp,  
\end{equation*}
\end{linenomath*}
with $R_a^{cl}\ge R_a^2$ and $R_\delta^{cl}\ge R_\delta^2$.

\begin{comment}
\section{Data Availability}
The MNIST and IRIS data sets used in the numerical simulations are available online at \url{http://yann.lecun.com/exdb/mnist/} and \url{https://archive.ics.uci.edu/ml/datasets/iris} respectively.

\section{Code Availability}
The python code used to generate the figures in this article are available from the authors on reasonable request.
\end{comment}

%%%%%%%%%%%%%%%%%%%%%%%%%%%%%%%%%%%%%%%%%%%%%%%%%%%%%%%%%%%%%%%%%%%%%%%%%%%%%%%%%%%%%%%%%%

\newpage
\bibliographystyle{unsrt}
\bibliography{bibliography.bib}

\begin{comment}
\section{Author Contributions}
J.A. and I.K. contributed to the ideas, calculations, and writing of the manuscript. J.A. additionally performed the numerical simulations. S.Z. contributed to the ideas and editing of the manuscript. C.H. contributed to discussions and provided input into the manuscript.

\section{Competing Interests}
The authors declare that there are no competing interests.

\end{comment}

%\appendix
\newpage
\pagenumbering{gobble}
\section{Supplemental Material}
\subsection{Classically sampling from the RIPE distribution}\label{app:ripe}

The quantum $(\epsilon,\gamma)$-feedforward and quantum $(\epsilon,\gamma)$-backpropagation subroutines (subroutines  \ref{sub:q-ff}  and \ref{sub:q-bp} respectively) of the main text make use of the Robust Inner Produce Estimation (RIPE) procedure to compute values $s$ that satisfy
$$ \abs{s - \langle x,y\rangle} \le \max\{\epsilon\abs{\langle x,y}\rangle,\epsilon\}$$
with probability at least $1-\gamma$, for given input vectors $x, y\in\mb{R}^n$. In this section we give a classical subroutine for generating samples from RIPE distribution, for the special case of RIPE where one has exact knowledge of the the norms $\norm{x}$ and $\norm{y}$ (the generalization to the case where we only have estimates of the norms is straightforward).  At a high level, this procedure is based on the following ideas (see \cite{kerenidis2018qmeans} for more details):

\begin{itemize}
\item The inner product estimation procedure of \cite{kerenidis2018qmeans} on which RIPE is based implicitly assumes access to a unitary operator for efficiently creating the state $\ket{\phi} = \frac{\norm{x}\ket{0}\ket{x} + \norm{y}\ket{1}\ket{y}}{\sqrt{\norm{x}^2 +\norm{y}^2}}$. Applying a Hadamard operator to the first register produces the state
$$ \ket{\Psi} =  \sqrt{a}\ket{1}\ket{\Psi_1} + \sqrt{1-a}\ket{0}\ket{\Psi_0} $$

where $a = \frac{\norm{x}^2 + \norm{y}^2 - 2\langle x,y \rangle}{2\lp\norm{x}^2 + \norm{y}^2\rp}$. Obtaining an estimate $\bar{a}$ to $a$ satisfying $\abs{\bar{a} - a} \le \epsilon_a = \frac{\epsilon \max\{1, \abs{\langle x,y\rangle}\}}{\norm{x}^2 + \norm{y}^2}$ will therefore allow us to compute an estimate $s$ to $\langle x,y\rangle$ satisfying $\abs{s - \langle x,y\rangle} \le \epsilon\max\{1,\abs{\langle x,y\rangle}\}$, by taking $s= \lp\norm{x}^2 + \norm{y}^2\rp\lp 1 - 2\bar{a}\rp/2$.
\item Performing amplitude estimation \cite{brassard2002quantum} on $\ket{\Psi}$ with $\log M$ ancilla qubits returns a value $\tilde{a}$ satisfying $\abs{\tilde{a} - a}\le \frac{\pi}{M} + \lp\frac{\pi}{M}\rp^2$ with probability at least $8/\pi^2$. This process requires time $\tilde{O}\lp M T\rp$, where $T$ is the time required to implement the unitary required for the creation of state $\ket{\phi}$.  Taking $M = \left\lceil\frac{\pi}{2\epsilon_a}\lp 1 + \sqrt{1 + 4\epsilon_a}\rp\right\rceil$
suffices to ensure that $\abs{\tilde{a} - a}\le \epsilon_a$.
\item By the Hoeffding bound, repeating the above procedure $Q = \left\lceil\frac{\log(1/\gamma)}{2(8/\pi^2-1/2)^2}\right\rceil_{odd}$ times and taking the median of the results gives a value $\bar{a}$ satisfying $\abs{\bar{a} - a}\le \epsilon_a $ with probability at least $1-\gamma$. The notation $\left\lceil z \right\rceil_{odd}$ denotes the smallest odd integer greater than or equal to $z$.
\end{itemize}

The amplitude amplification part of the sampling subroutine makes use of a distance function $d:\mb{R}^2\ra \mb{Z}$ given by $d(\omega_0,\omega_1) = \min_{z\in\mb{Z}} \{z + \omega_1 - \omega_0\}$.

\begin{subroutine}(\textbf{Classically sampling from the RIPE distribution}) \label{sub:classical-ipe}\\

\noindent
{\bf Inputs}:  $x,y\in\mb{R}^n$,  $\epsilon, \gamma >0$, $Q = \left\lceil\frac{\log(1/\gamma)}{2(8/\pi^2-1/2)^2}\right\rceil_{odd}$.
\begin{enumerate}
    \item Compute $a = \frac{\norm{x}^2 + \norm{y}^2 - 2\langle x,y \rangle}{2\lp\norm{x}^2 + \norm{y}^2\rp}$, $\theta_a = \sin^{-1}\lp\sqrt{a}\rp$,  $\epsilon_a = \frac{\epsilon \max\{1, \abs{\langle x,y\rangle}\}}{\norm{x}^2 + \norm{y}^2}$, $M = \left\lceil\frac{\pi}{2\epsilon_a}\lp 1 + \sqrt{1 + 4\epsilon_a}\rp\right\rceil$
    \item For $j=1$ to $M$ do:
    \item \qquad $a_j = \sin^2\lp\pi j / M\rp$
    \item \qquad $p(a_j) = \abs{\frac{\sin\lp M d\lp j/M, \theta_a/\pi\rp\rp}{M\sin\lp d\lp j/M, \theta_a/\pi\rp\rp}}^2$ 
    \item For $q = 1$ to $Q$ do:
    \item \qquad  Sample $\tilde{a}^{(q)} \sim p$ 
    \item Compute $\bar{a} = \mathsf{median}(\tilde{a}^{(1)},\ldots \tilde{a}^{(Q)})$
    \item Return $s = \lp\norm{x}^2 + \norm{y}^2\rp(1-2\bar{a})/2$
\end{enumerate}
\end{subroutine}

Examples of such samples are shown in Fig \ref{fig:samples} for randomly chosen vectors $x$ and $y$ given by
\begin{align}
 x &= (5.88414114, 2.0327562 , 1.68155901, 7.91848042, 1.61922687),   \label{eq:x_y}\\
 y &=(5.15610287, 7.2034771 , 9.88496245, 3.46281654, 4.20607662), \nonumber
\end{align}
$\epsilon = 0.3$ and various values of $\gamma$.  These vectors have norms $\norm{x} = 10.34$ and $\norm{y} = 14.35$, and inner product $\langle x, y \rangle = 95.38$.  
\begin{figure}[h]%
\subfloat[$\gamma = 0.2$, $Q=1$, $f = 86\%$.]{\includegraphics[height=4.5cm]{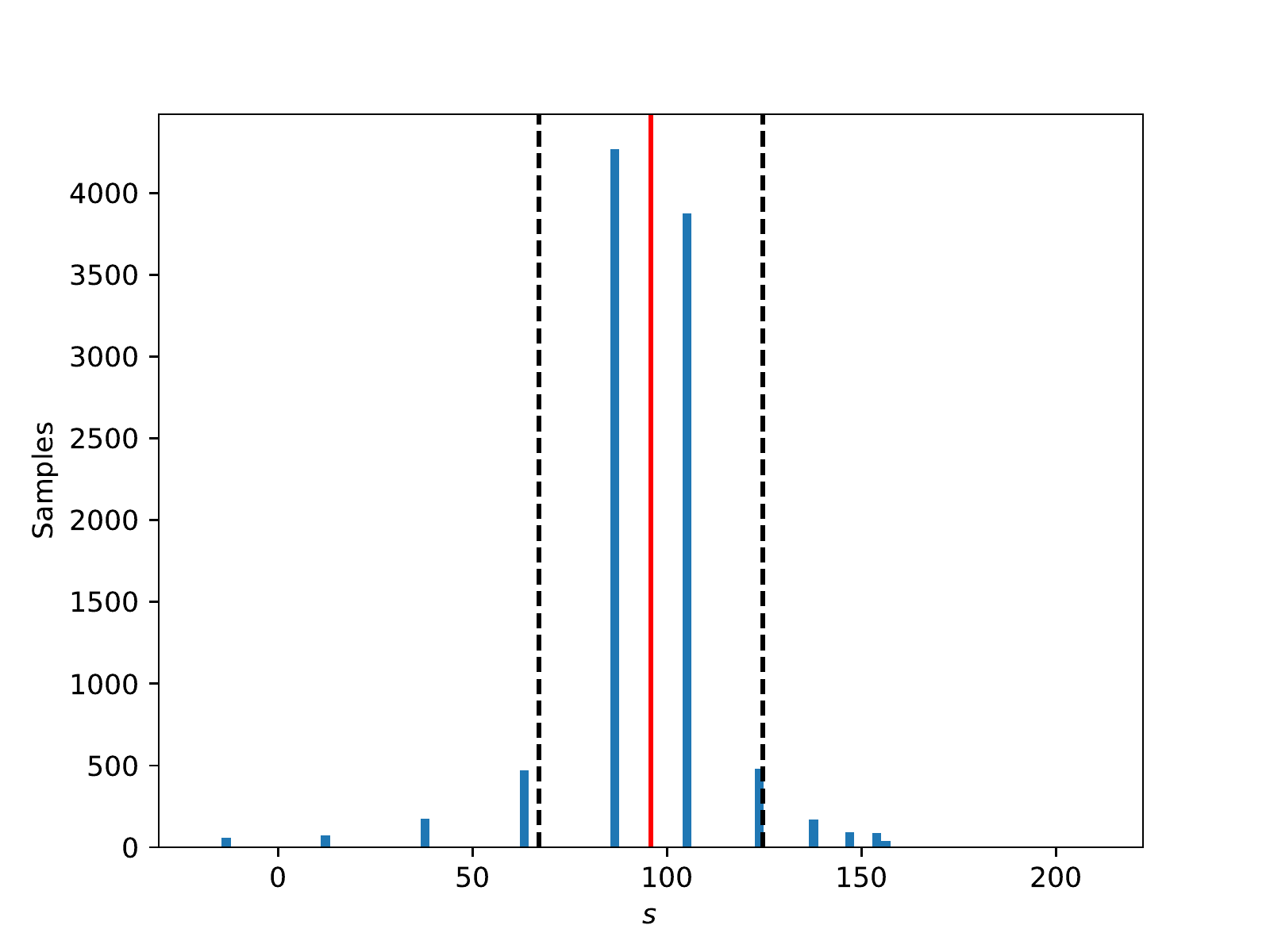}}
\subfloat[$\gamma = 0.05$, $Q = 3$, $f = 97\%$.]{\includegraphics[height=4.5cm]{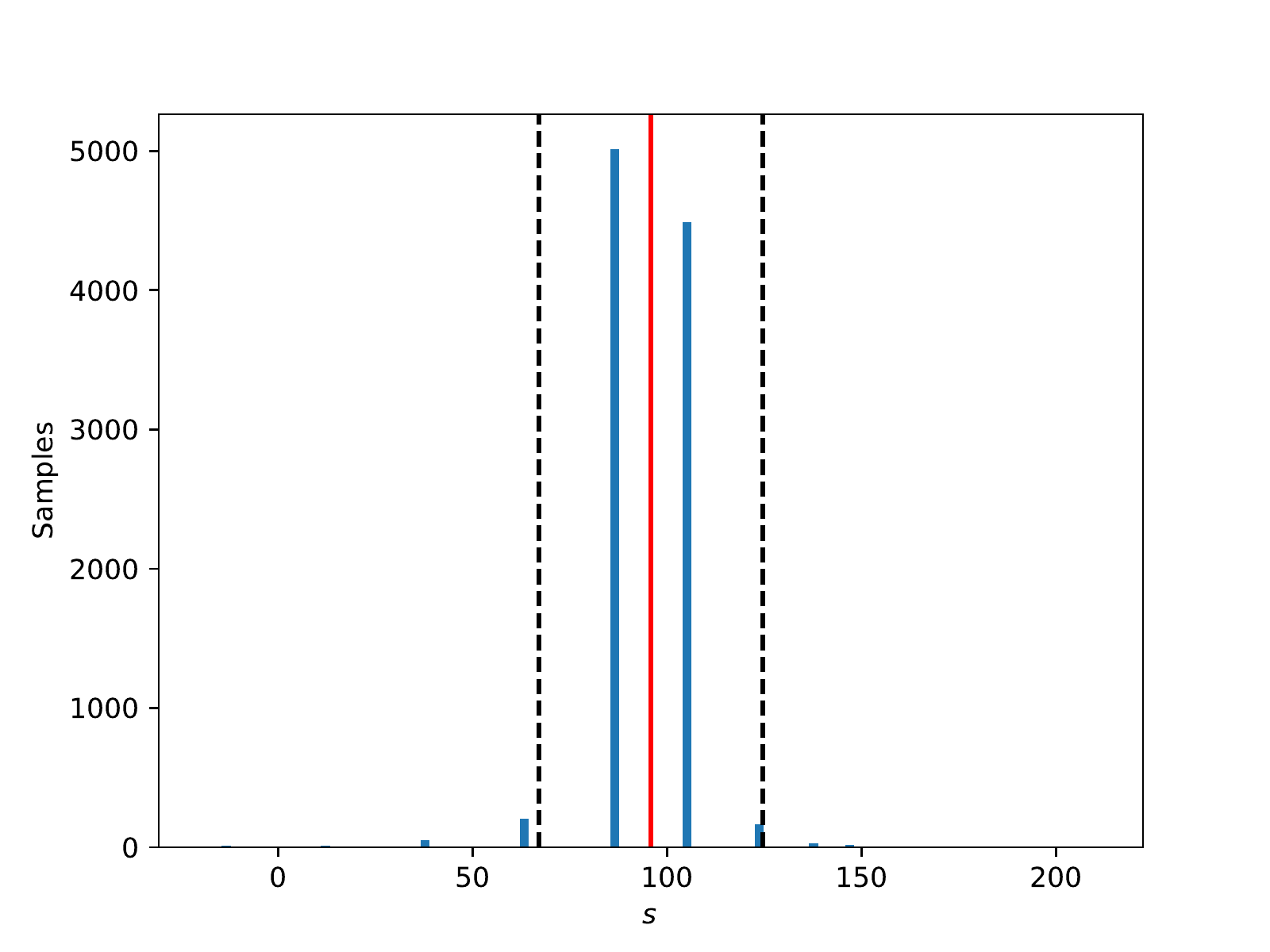}}
\subfloat[$\gamma = 0.01$,$Q = 5$, $f = 99\%$.]{\includegraphics[height=4.5cm]{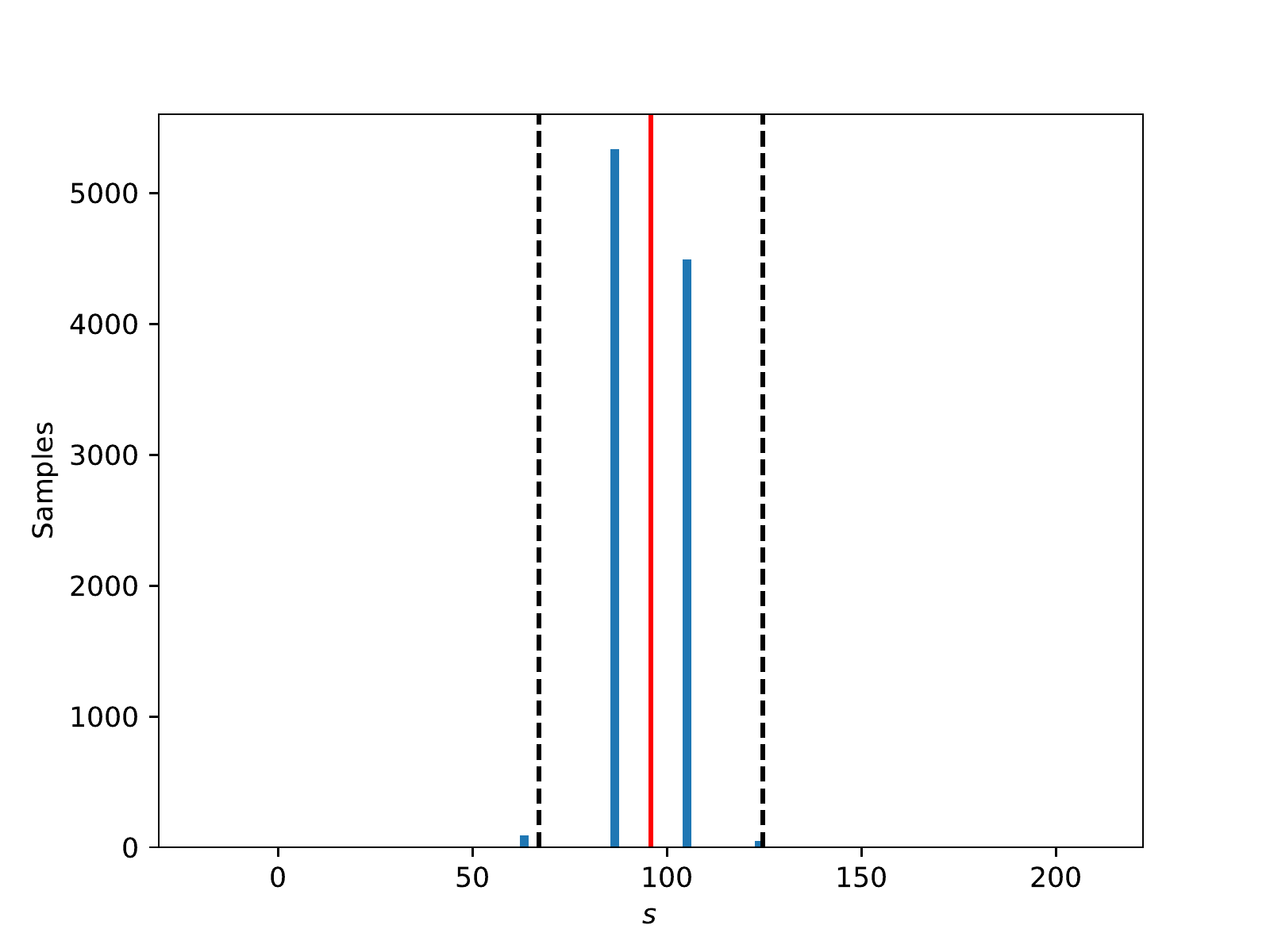}}
\caption{Concentration of $10,000$ samples of estimates $s$ of the inner product of vectors $x,y$ given in \eqref{eq:x_y}, drawn from the RIPE distribution, for $\epsilon = 0.3$ and various values of $\gamma$. The vertical red line indicates the true value of the inner product $\langle x, y\rangle = 95.83$, and the vertical dashed lines are located at $\pm \epsilon\max\{1,\abs{\langle x,y \rangle}\}$. The actual percentage of points sampled which lay within the desired range in each case is denoted by $f$. Note that in the case of $(a)$ $\gamma = 0.2$ corresponds to $1 - \gamma < 8/\pi^2$, so $Q=1$ suffices and there is no need to take medians.}
\label{fig:samples}%
\end{figure}

\end{document}